\documentclass[journal]{IEEEtran} 
 
\normalsize
 
\usepackage{cite}
\usepackage[]{graphicx}
\graphicspath{{figures/}}
\DeclareGraphicsExtensions{.eps,.pdf}
\usepackage[mathscr]{euscript}
\usepackage{balance}

\usepackage{dsfont}
\usepackage{hyperref}
\usepackage[usenames,dvipsnames]{color}

\usepackage{graphicx}
\usepackage{subcaption} 
\usepackage{latexsym}
\usepackage{amssymb}
\usepackage{verbatim}
\usepackage{amsmath}
\usepackage{amsthm}
\usepackage{amsfonts}
\usepackage{stackengine}
\usepackage{bm} 
\usepackage{float}
\usepackage{stfloats}
\usepackage{mathtools}
\usepackage{dsfont}

\usepackage{multicol}

\newtheorem{Theorem}{Theorem}

\newtheorem{Remark}{Remark}

\newtheorem{Example}{Example}

\IEEEoverridecommandlockouts

\begin{document}
%
\title{Private and Secure Distributed Matrix Multiplication with Flexible Communication Load}
%
%
%

\author{\IEEEauthorblockN{\textbf{Malihe Aliasgari}\IEEEauthorrefmark{1},
\textbf{Osvaldo Simeone}\IEEEauthorrefmark{2}
and \textbf{J\"org Kliewer}\IEEEauthorrefmark{1}} \\
\IEEEauthorrefmark{1} Department of Electrical and Computer
Engineering, New Jersey Institute of Technology,  Newark, NJ, U.S.A.\\
\IEEEauthorrefmark{2} Department of Engineering, 
King's College London, Department of Engineering,  London, U.K.
\thanks{This work was supported in part by 
the European Research Council (ERC) under
the European Union Horizon 2020 research and innovative programme (grant
agreement No 725731)
and by U.S. NSF grants CNS-1526547, CCF-1525629.

Part of this paper was presented at the IEEE International Symposium on Information Theory (ISIT), 2019 \cite{aliasgari2019distributed}.
}
\thanks{M. Aliasgari and J. Kliewer are with Helen and John C. Hartmann  Department of Electrical and Computer Engineering, New Jersey Institute of Technology, Newark, New Jersey, USA (email: ma839@njit.edu; jkliewer@njit.edu).}
\thanks{O. Simeone is with King's Centre for Learning \& Information Processing (kclip), Centre for Telecommunication Research (CTR), the Department of Engineering, King's College London, London, UK (email: osvaldo.simeone@kcl.ac.uk).}
}

\maketitle

\begin{abstract} 
Large matrix multiplications are central to large-scale machine learning applications. These operations are often carried out on a distributed computing platform with a master server and multiple workers in the cloud operating in parallel. For such distributed platforms, it has been recently shown that coding over the input data matrices can reduce the computational delay, yielding a trade-off between recovery threshold, i.e., the number of workers required to recover the matrix product, and communication load, i.e., the total amount of data to be downloaded from the workers. In this paper, in addition to exact recovery requirements, we impose security and privacy constraints on the data matrices, and study the recovery threshold as a function of the communication load. We first assume that both matrices contain private information and that workers can collude to eavesdrop on the content of these data matrices. For this problem, we introduce  a novel class of secure codes, referred to as secure generalized PolyDot (SGPD) codes, that generalize state-of-the-art non-secure codes for matrix multiplication. SGPD codes allow a flexible trade-off between recovery threshold and communication load for a fixed maximum number of colluding workers while providing perfect secrecy for the two data matrices. 
We then study a connection between secure matrix multiplication and private information retrieval. 
We specifically assume that one of the data matrices is taken from a public set known to all the workers. In this setup, the identity of the matrix of interest should be kept private from the workers. For this model, we present a variant of generalized PolyDot codes that can guarantee both secrecy of one matrix and privacy for the identity of the other matrix for the case of no colluding servers.
\end{abstract}
 

 \IEEEpeerreviewmaketitle

\section{Introduction}
\subsection{Motivation and Problem Definition}
At the core of many signal processing and machine learning applications are
tensor operations, most notably large matrix multiplications
\cite{janzamin2015beating}. In the presence of practically sized data sets,
such operations are typically carried out using distributed computing
platforms with a master server and multiple workers that can operate in
parallel over distinct parts of the data set. The master server plays the
role of the parameter server, distributing data to the workers and periodically reconciling their internal state \cite{li2014scaling}. Workers are commercial off-the-shelf servers that are characterized by possible temporary failures and delays \cite{dean2013tail}.

Straggling workers can affect the computation latency by orders of magnitude, e.g., \cite{joshi2017efficient,wang2015using}. While current distributed computing platforms conventionally handle straggling servers by means of replication of computing tasks \cite{huang1984algorithm}, recent work has shown that encoding the input data can help reduce the computation latency.
More generally, coding is able to  control the trade-off between
computational delay and communication load between workers and master server
\cite{lee2018speeding,yu2017polynomial,li2018fundamental,aliasgari2018coded,aliasgari2018codedisit,dutta2018optimal,dutta2018unified,fahim2017optimal,fahim2019numerically,subramaniam2019random}.
Furthermore, stochastic coding can help keeping both input and output data
secure from the workers, assuming that the latter are honest, i.e., carrying out the prescribed protocol, but curious \cite{nodehi2018limited,yu2018lagrange,chang2018capacity,kakar2019capacity,yang2019secure,rafael2018codes,das2019random,nodehi2019secure}. 
This paper contributes to this line of work by investigating the trade-off between computational delay and communication load as a function of the privacy level.

As illustrated in Figs.~\ref{FigComputSyst} and \ref{FigSyst}, we focus on the basic problem of computing a matrix multiplication $\mathbf{C}=\mathbf{AB}$ in a distributed computing system of $P$ workers that can process each only a fraction $1/m$ and $1/n$ of matrices $\mathbf{A}$ and $\mathbf{B}$, respectively.
In the first setup under study, illustrated in Fig.~\ref{FigComputSyst}, both matrices $\mathbf A$ and $\mathbf B$ are to be kept private from the workers. Here, three performance criteria are of interest: 
\begin{itemize}
	\item the recovery threshold $P_R$, that is, the number of workers that need to complete their task before the master server can recover the product $\mathbf{C}$;
	\item the communication load $C_L$ between workers and master server, i.e., the amount of information to be downloaded from the workers;
	\item the maximum number $P_C$ of colluding servers that ensures perfect secrecy for both data matrices $\mathbf{A}$ and $\mathbf{B}$.
\end{itemize}
In the second setup of interest shown in Fig.~\ref{FigSyst}, only matrix $\mathbf A$ is private, while matrix $\mathbf B$ is selected from a public data set $\mathcal B$. In this case, apart from the security constraint on $\mathbf A$, we only impose a privacy constraint on the identity of the specific matrix $\mathbf B\in \mathcal B$ of interest. 
As a motivation for this second setup, consider a recommender system based on collaborative filtering  \cite{ricci2011introduction}.
In this case, recommendations are based on the product of two matrices, one describing the profile of a user, or a group of users, and one representing features of the items of interest, such as movies, music or TV shows. The users' profile matrix can be modelled by the private matrix $\mathbf A$, hence ensuring the privacy of users' data; while the items' data matrix for each category is represented by one of the matrices in the public data set $\mathcal B=\{\mathbf B^{(k)}\}_{k=1}^L$. This latter assumption captures the constraint that users may want to keep the confidential types of items they are interested in. 
For this problem, the criteria of interest are still $P_R$ and $P_C$, and we simplify the problem by setting $P_C=1$. This paper focuses on the design of coding and computing techniques for both problems.

\subsection{Related Work}
In order to put our contribution in perspective, we briefly review prior related work.
Consider first solutions that provide no security guarantees, i.e.,
$P_C=0$, for the problem in Fig.~\ref{FigComputSyst}. As a direct extension of \cite{lee2018speeding}, a first approach
is to use product codes that apply separately the maximum  distance  separable (MDS) codes to encode the two
matrices \cite{lee2017high}. The recovery threshold of this scheme is
improved by \cite{yu2017polynomial}, which introduces \textit{polynomial codes}.
The construction in \cite{yu2017polynomial} is proved to be optimal under the assumption that \textit{minimal communication} is allowed between workers and master server. In \cite{fahim2017optimal}, MatDot codes are introduced, resulting in a lower recovery threshold at the expense of a larger communication load. The construction in \cite{dutta2018optimal} bridges the gap between polynomial and MatDot codes and presents PolyDot codes, yielding a trade-off between recovery threshold and communication load. An extension of this scheme, termed Generalized PolyDot
(GPD) codes improves on the recovery threshold of PolyDot codes
\cite{dutta2018unified}, which is independently obtained also by the
construction~in~\cite{yu2018straggler}. In [14], GPD codes are used to design a unified coded computing strategy for the training of deep neural networks.

Much less work has been done in the
literature for the case in which security constraints are factored in, i.e., where $P_C\ne 0$, for the problem of Fig.~\ref{FigComputSyst}. In
\cite{yu2018lagrange}, Lagrange coding is presented that achieves the
minimum recovery threshold for multilinear functions by generalizing MatDot codes. In \cite{nodehi2018limited,nodehi2019secure}, coded schemes have been used to develop multi-party computation
techniques to calculate arbitrary polynomials of massive matrices, preserving the security of the data matrices.
In \cite{chang2018capacity,kakar2019capacity,rafael2018codes} a reduction
of the communication load is obtained by extending polynomial codes. While
these works focus on either minimizing recovery threshold or communication
load, the \emph{trade-off} between these two fundamental quantities has not been addressed in the open literature to the best of our knowledge. A new class of secure distributed matrix multiplication and its capacity is studied in \cite{jia2019capacity}.

In the second part of this work, we study a connection between secure matrix multiplication and private information retrieval (PIR), as illustrated in Fig.~\ref{FigSyst}. The PIR problem was introduced in \cite{chor1995private} and has been widely studied in recent years, e.g., in \cite{gasarch2004survey,yekhanin2010private,sun2017capacity,banawan2018capacity,Hollanti_etal_2017,kazemi2019single,kazemi2019private,kim2019private,chang2019upload,tahmasebi2019private}.
In \cite{kim2019private} and \cite{chang2019upload} the PIR setup was investigated for the problem of distributed matrix multiplication  illustrated in Fig.~\ref{FigSyst} that imposes PIR guarantees for the index of matrix $\mathbf B$ within a public library. In \cite{kim2019private}, a coding strategy is proposed that combines the PIR scheme for non-colluding servers (i.e., with $P_C=1$) \cite{chor1995private} with polynomial codes \cite{yu2017polynomial}. In \cite{chang2019upload}, the authors introduce a related approach for this problem, and show that it outperforms the scheme proposed in \cite{kim2019private} in terms of upload and download cost. The code design in \cite{chang2019upload} focuses on the minimization of the communication load, and does not explore the trade-off between this metric and the recovery threshold.   

\subsection{Main Contribution}
In this paper we first present a novel class of secure computation codes, referred to as secure GPD (SGPD) codes, for the setup in Fig.~\ref{FigComputSyst}, SGPD codes generalize GPD codes to operate at a flexible communication load level.  This yields a new achievable trade-off between
recovery threshold $P_R$ and communication load $C_L$ as a function of a prescribed number of colluding workers $P_C$. 
In the process, we also introduce a novel perspective on distributed computing codes based on the signal processing concepts of convolution and $z$-transform. 
SGPD codes were first introduced in the conference version of this paper \cite{aliasgari2019distributed}, which did not contain complete proofs and provided only limited illustrations and examples.
Then, SGPD codes are modified to offer a solution, introduced here for the first time, for the scenario in Fig.~\ref{FigSyst}. This is done through concatenation with the PIR code in \cite{kim2019private}, which ensures both secrecy of the input matrix $\mathbf A$ and privacy of the identity for the desired matrix in the library $\mathcal B$ if $P_C=1$.
The resulting codes are referred to as private and secure GPD (PSGPD) codes. They generalize the approach in \cite{chang2019upload}, enabling a trade-off between (upload) communication load and recovery threshold. 
We finally illustrate the benefits of the proposed codes, which offer a flexible trade-off between communication load and recovery threshold, by analyzing the overall completion time due to both computation and communication.

\subsection{Organization}
The rest of the paper is organized as follows. In Section \ref{secModel}, we present the system models for secure matrix multiplication (Fig.~\ref{FigComputSyst} in Section \ref{subsec1_sec}) and for private and secure matrix multiplication (Fig.~\ref{FigSyst} in Section \ref{Sec_PrivateSecure}), respectively.
In Section \ref{Sec_GPD} we propose an intuitive interpretation of the GPD code introduced in \cite{fahim2017optimal}. Using $z$-transforms, Section \ref{sec_SecPDC} proposes a novel extension of GPD codes by imposing a security constraint on the data matrices and deriving the resulting trade-off between recovery threshold $P_R$ and communication load $C_L$. In this section, we also study overall completion latency encompassing both computation and communication latencies for SGPD codes.
In Section \ref{Sec_sec_Pri}, we address the setup in Fig.~\ref{FigSyst}, again with respect to the trade-off between $P_R$ and $C_L$ and to the overall completion latency. 
The paper is concluded in Section \ref{Sec_Con}.

\section{Problem Statement}\label{secModel}
\subsection{Notation}
Throughout the paper, we denote a matrix with upper boldface letters (e.g., $\mathbf{X}$), and lower boldface letters indicate a vector or a sequence of matrices (e.g., $\mathbf{x}$).
Furthermore, a math calligraphic font refers to a set (e.g., $\mathcal{X})$.
A set $\mathbb{F}$ represents the Galois field with cardinality $|\mathbb{F}|$. We denote by $\mathbb N$ the set of all non-zero positive integers, and for some $a, b \in \mathbb N$, $a\leq b$, $[a , b]\overset{\Delta}{=}\{a, a+ 1, \ldots, b\}$.
For any real number $r$, $\lceil r\rceil$ represents the largest integer nearest to $r$.
he function $H(\cdot)$ represents the entropy of its argument, and $I(X ;Y)$ denotes the mutual information of the random variables $X$ and $Y$.

 \begin{figure*}[t!]
	\begin{center}
	\includegraphics[scale=0.28]{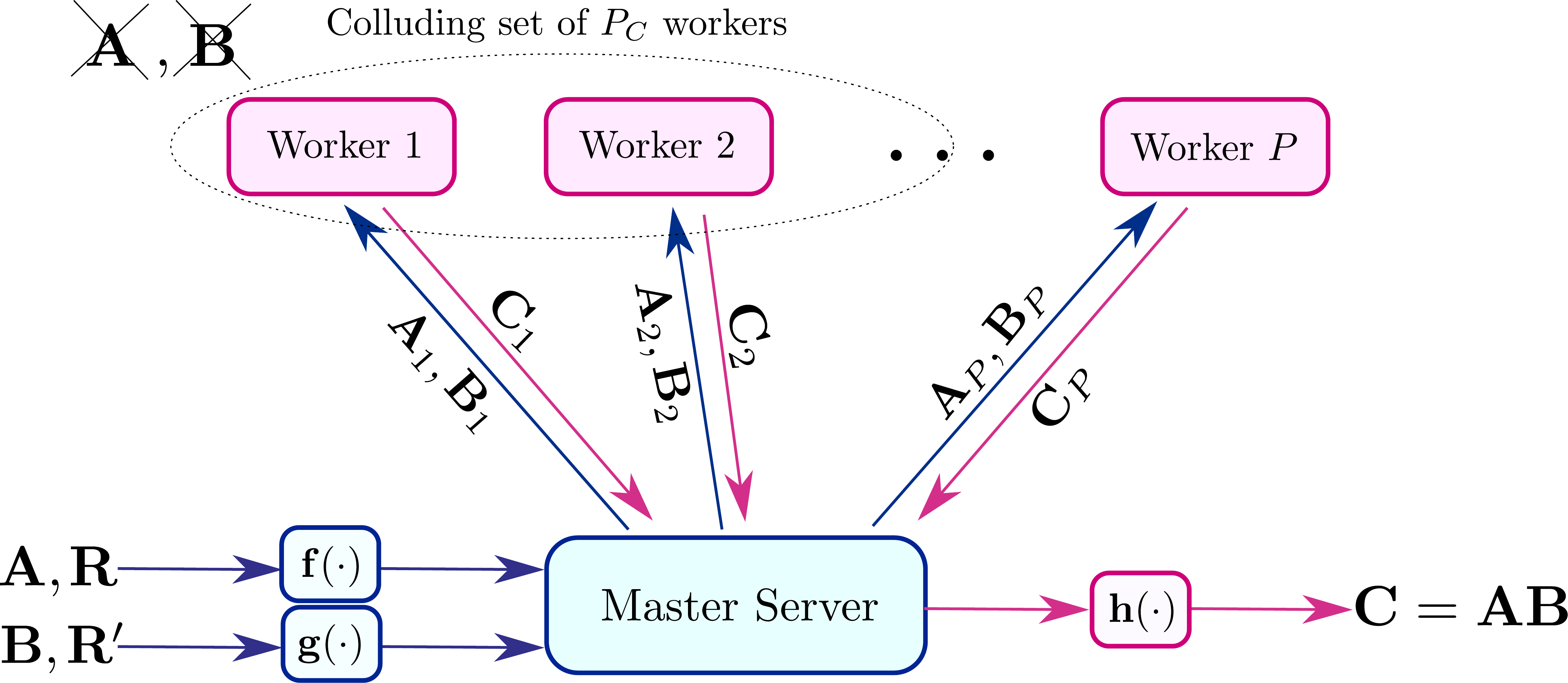}\vspace{-.1cm}~\caption{\footnotesize{ Secure matrix multiplication: the master server encodes both input matrices $\mathbf A$ and $\mathbf B$, to be kept secure from the workers, and both random matrices $\mathbf{R}$ and $\mathbf{R}'$, respectively, to define the computational tasks of the slave servers or workers. The workers may fail or straggle, and they are honest but curious, with colluding subsets of workers of size at most $P_C$. The master server must be able to decode the product $\mathbf C=\mathbf{AB}$ from the output of a subset of $P_R$ servers, which defines the recovery threshold.}}~\label{FigComputSyst}
	\end{center}
	\vspace{-5ex}
\end{figure*}

\subsection{System Model}
As illustrated in Figs.~\ref{FigComputSyst} and \ref{FigSyst}, we consider a distributed computing system with a master server and $P$ slave servers or workers. The master server is interested in computing securely the matrix product $\mathbf C=\mathbf{AB}$ of two data matrices $\mathbf{A}$ and $\mathbf{B}$ with dimensions $T\times S$ and $S\times D$, respectively. The matrices have i.i.d. uniformly distributed  entries from a sufficient large finite field $\mathbb{F}$, with $|\mathbb{F}|>P$.
More precisely, we will consider two scenarios. In the first, both matrices $\mathbf{A}$ and $\mathbf{B}$ are available at the master server and contain confidential data that should be kept secure from the workers (see Fig.~\ref{FigComputSyst}). In the second, only matrix $\mathbf A$ contains confidential information, and there are $L$ public matrices in the set $\mathcal B=\{\mathbf B^{(r)}\}_{r=1}^L$ from which the master node wishes to compute the product $\mathbf C^{(\kappa)}=\mathbf{AB}^{(\kappa)}$ for some $\kappa$th index $\kappa\in[1,L]$. The index must be kept private against the workers (see Fig.~\ref{FigSyst}).
In the following, we first describe the system model for the setup in Fig.~\ref{FigComputSyst}, referred to as \textit{secure matrix multiplication}, followed by the setup for the model in Fig.~\ref{FigSyst}, referred to as \textit{private and secure matrix multiplication}.

\subsection{Secure Matrix Multiplication}\label{subsec1_sec}
For the scenario in Fig.~\ref{FigComputSyst} workers receive information on matrices $\mathbf{A}\in \mathbb F^{T\times S}$ and $\mathbf{B}\in \mathbb F^{S\times D}$ from the master server; they process this information and they respond to the master server, which finally recovers the product $\mathbf C=\mathbf{AB}$ with minimal computational effort. Due to communication and complexity constraints, each worker can receive only $TS/m$ and $SD/n$ symbols, respectively, for some integers $m$ and $n$. 
The workers are honest but curious. Accordingly, we impose the secrecy constraint that, even if up to $P_C < P$ workers collude, the workers cannot obtain any information about both matrices $\mathbf A$ and $\mathbf B$  based on the data received from the master server. 

To keep the data secure and to leverage possible computational redundancy at the workers (namely, if $P/m>1$ and/or $P/n>1$), the master server sends encoded versions of the input matrices to the workers due to the above mentioned communication and complexity constraints.
Specifically, it produces the encoded matrices $ {\mathbf A}_p=\mathbf f_{p}(\mathbf A,\mathbf{R})$, where $\mathbf R$ is a random matrix of dimension $T'\times S'$, for some integers $T'$ and $S'$ to be defined below, via the function 
\begin{equation}\label{Ap}
\mathbf f_{p}:\mathbb{F}^{T\times S}\times \mathbb{F}^{T'\times S'}\rightarrow\mathbb{F}^{T/t\times S/s},
\end{equation}
for some integers $t$ and $s$ such that $m=st$. The resulting $TS/m$ entries in the output of function $\mathbf f_p$ are then sent to worker $p$, with $p\in[1,P]$.
Likewise, the master server computes the encoded matrices $ {\mathbf B}_p=\mathbf g_{p}(\mathbf B,\mathbf{R}')$, where $\mathbf R'$ is a random matrix  of dimension $S'\times D'$, for some integers $S'$ and $D'$ to be defined below, using the function
\begin{equation}\label{Bp}
\mathbf g_{p}:\mathbb{F}^{S\times D}\times \mathbb{F}^{S'\times D'}\rightarrow\mathbb{F}^{S/s\times D/d},
\end{equation} 
for some integers $s$ and $d$ such that $n=sd$. The resulting $SD/n$ entries in $\mathbf B_p$ are then sent to worker $p$. The random matrices $\mathbf{R}$ and $\mathbf{R}'$ consists of i.i.d. uniformly distributed entries from a field $\mathbb{F}$. 
The security constraint imposes the condition
\begin{equation}\label{eq_security_constraint}
I( {\mathbf A}_{\mathcal{P}}, {\mathbf B}_{\mathcal{P}};\mathbf A,\mathbf B)=0, 
\end{equation}
for all subsets of $\mathcal{P}\subset [1,P]$ of $P_C$ workers, where the random matrices $\mathbf{R}$ and $\mathbf{R}'$ serve as random keys in order to meet the security constraint \eqref{eq_security_constraint} \cite{shamir1979share}.

Each worker $p$ computes the product $\mathbf C_p= {\mathbf A}_p {\mathbf B}_p$ of the encoded sub-matrices $ {\mathbf A}_p$ and $ {\mathbf B}_p$. 
The master server collects a subset of $P_R\leq P$ outputs from the workers as defined by the subset $\{\mathbf{C}_p\}_{p \in \mathcal{P}_R}$ with $|\mathcal{P}_R|=P_R$. It then applies a decoding function as
$\mathbf{h}\left(\{\mathbf C_p\}_{p\in \mathcal{P}_R}\right)$,
\begin{equation}\label{eq_encodedC} \mathbf{h}:\underbrace{\mathbb{F}^{T/t\times D/d}\times\cdots \times \mathbb{F}^{T/t\times D/d}}_{P_R \text{~times}}\rightarrow \mathbb{F}^{T\times D}. 
\end{equation}
Note that \textit{correct decoding} translates into the condition 
\begin{equation}\label{eq_decidability_constraint}
H(\mathbf{AB}|\{\mathbf C_p\}_{p\in \mathcal{P}_R})=0. 
\end{equation}
A coding and decoding strategy that satisfies condition \eqref{eq_security_constraint} and \eqref{eq_decidability_constraint} is said to be \textit{feasible}.

For given  parameters $m$ and $n$ the performance of a coding and decoding scheme is measured by the triple $(P_C,P_R,C_L)$,
where $C_L$ is defined as 
\begin{equation}\label{CommLoad}
C_L=\sum_{p\in \mathcal{P}_R}|\mathbf{C}_p|;
\end{equation}
$|\mathbf{C}_p|$ is the dimension of the product matrix $\mathbf{C}_p$ computed by worker $p$.
Note that condition \eqref{eq_decidability_constraint} requires  the inequality 
$\min\{P_R/m,P_R/n\}\geq 1$ or
$P_R\geq P_{R,\min} \overset{\Delta}{=} \max \{m,n\}$, which is hence a lower bound for the minimum recovery threshold. Furthermore, the communication load is lower bounded by $C_L\geq C_{L,\min}\overset{\Delta}{=}TD$, which is the size of the product $\mathbf{C}=\mathbf{AB}$.
 
\begin{figure*}[t!]
	\begin{center}
		\includegraphics[scale=0.268]{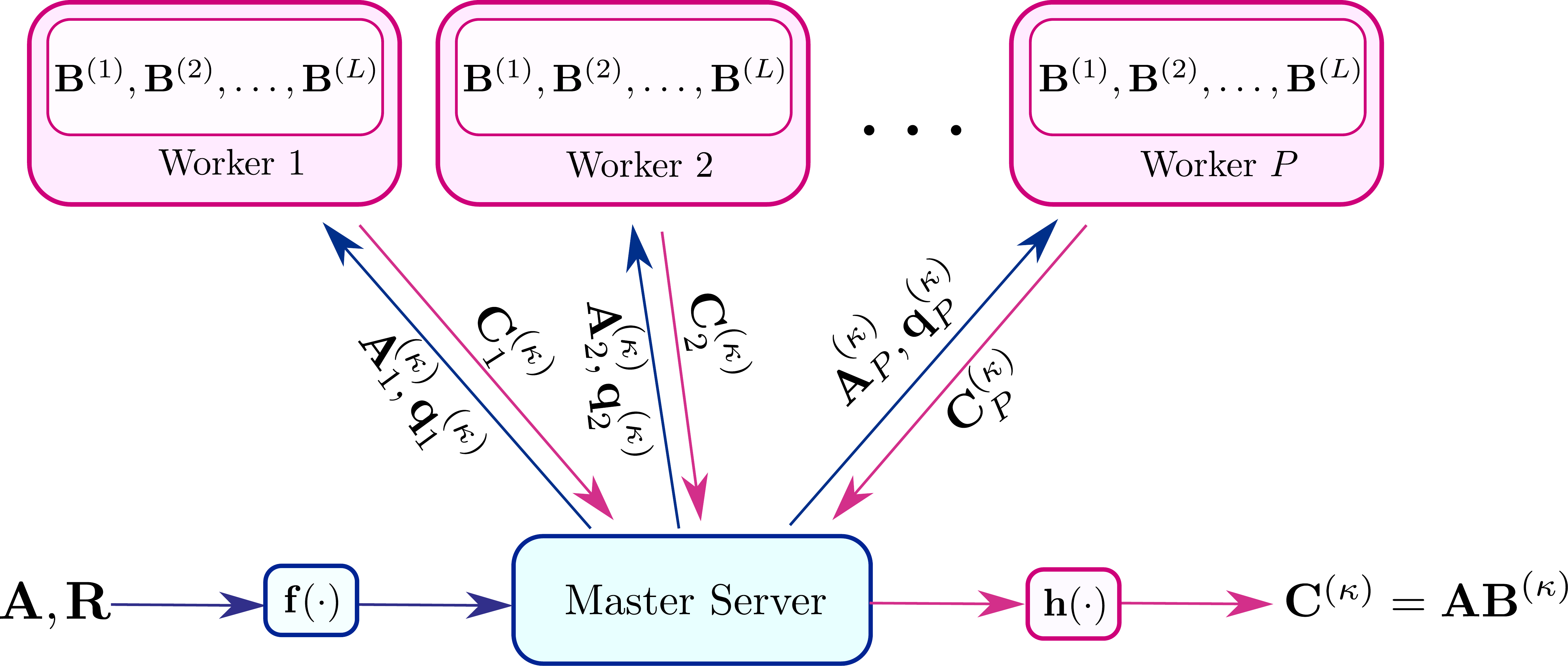}\vspace{-.1cm}~\caption{\footnotesize{Private and secure matrix multiplication: the master server encodes the input matrix $\mathbf A$, to be kept secret from the workers, and generates the encoded matrix $\mathbf A_p^{(\kappa)}$ for each worker $p$. It also sends a query $\mathbf q_p^{(\kappa)}$ as a function of the index $\kappa\in [1,L]$, to be kept private from workers, of the desired product $\mathbf C^{(\kappa)}=\mathbf{AB}^{(\kappa)}$, with matrices $\{\mathbf{B}^{(r)}\}_{r=1}^L$ available at all workers. The non-colluding workers may fail or straggle, and they are honest but curious. The master server must be able to decode the product $\mathbf{C}^{(\kappa)}$ from the output of a subset of $P_R$ servers, which defines the recovery threshold.}}~\label{FigSyst}
	\end{center}
	\vspace{-5ex}
\end{figure*}
 
\subsection{Private and Secure Matrix Multiplication}\label{Sec_PrivateSecure}

In this subsection, we discuss the private and secure matrix multiplication problem illustrated in Fig.~\ref{FigSyst}.  
In this setup, the master server wishes to compute the product $\mathbf C^{(\kappa)}=\mathbf {AB}^{(\kappa)}$ of a confidential input matrix $\mathbf A$ with a matrix $\mathbf B^{(\kappa)}$ from a set of public matrices $\{\mathbf B^{(1)},\ldots,\mathbf B^{(L)}\}$, while keeping the index $\kappa$ of the matrix $\mathbf B^{(\kappa)}$ of interest private from the workers.

Similar to the secure model in Fig.~\ref{FigComputSyst}, we consider a distributed computing system with a master server and $P$ honest but curious workers.
The master server contains a confidential data matrix $\mathbf{A}$ with dimension $T\times S$. Each worker has access to the library $\mathcal{B}$, which consists of $L$ distinct matrices $\{\mathbf B^{(1)},\ldots,\mathbf B^{(L)}\}$, each with dimension $S\times D$. As above, all matrices contain data symbols chosen uniformly i.i.d. from a sufficient large finite field $\mathbb{F}$, with $|\mathbb{F}|>P$.
The master server is interested in computing the matrix product $\mathbf{C}^{(\kappa)}=\mathbf{AB}^{(\kappa)}$ of the data matrix $\mathbf{A}$ and of a matrix $\mathbf{B}^{(\kappa)}$ for some index ${\kappa}\in[1,L]$.
This should be done while keeping the data matrix $\mathbf{A}$ secret against the workers in the same sense as in the scenario of Fig.~\ref{FigComputSyst}, while also ensuring that the index $\kappa$ is kept secret from the workers. 

To do so, as in the PIR problem \cite{sun2017capacity,banawan2018capacity}, the master server generates $P$ query vectors
$\mathbf q_1^{(\kappa)},\ldots,\mathbf q_P^{(\kappa)}\in\mathbb F^L$, for some $L>1$ as a function of the desired index $\kappa$ and sends each worker $p\in[1,P]$, the query vector $\mathbf q_p^{(\kappa)}$.
We assume that the workers do not collude, i.e., we set $P_C=1$. Extensions to any $P_C>1$ are possible and are left for future work. We note that, when the input matrix $\mathbf A$ is an identity matrix, the setup reduces to a PIR problem.

To keep the data matrix $\mathbf{A}$ secure against workers, the master server sends each worker $p\in[1,P]$ an encoded version $\mathbf{A}_p^{(\kappa)}=\mathbf f_p(\kappa,\mathbf A,\mathbf R)\in \mathbb F^{T/t\times S/s}$ which is a function of index $\kappa$, and through it, of the query $\mathbf q_{p}^{(\kappa)}$, of the
data matrix $\mathbf{A}$ and of a random matrix $\mathbf{R}$, 
for some integers $t$ and $s$ such that $m=ts$.

Upon receiving $(\mathbf q_p^{(\kappa)},\mathbf A_p^{(\kappa)})$, each worker $p$ uses the query $\mathbf q_p^{(\kappa)}$ to derive an $S/s\times D/d$ matrix $\mathbf B_p^{(\kappa)}=\mathbf g_{p}(\mathbf q_p^{(\kappa)},\mathcal B)\in\mathbb F ^{S/s\times D/d}$ from the library $\mathcal B$ by using an encoding function 
\begin{equation}\label{Bp2}
\mathbf g_{p}:\mathbb{F}^L\times\underbrace{\mathbb{F}^{S\times D}\times\cdots\times \mathbb{F}^{S\times D}}_{L \text{ times}}\rightarrow\mathbb{F}^{S/s\times D/d},
\end{equation} 
for some integers $s$ and $d$ such that $n=sd$. We emphasize that, unlike the setup considered in Fig.~\ref{FigComputSyst}, the \textit{content} of the desired matrix $\mathbf B^{(\kappa)}$ is not secure against workers, since the library $\mathcal B$ is public.  
Each worker $p$ then computes the product $\mathbf C_p^{(\kappa)}=\mathbf A_p^{(\kappa)}\mathbf B_p^{(\kappa)}$ and sends it to the master server. The master server collects a subset $\{\mathbf{C}_p^{(\kappa)}\}_{p \in \mathcal{P}_R}$ of $P_R\leq P$ outputs from the workers  with $|\mathcal{P}_R|=P_R$. It then applies a decoding function $\mathbf{h}(\{\mathbf C_p^{(\kappa)}\}_{p\in \mathcal{P}_R})$, as in \eqref{eq_encodedC}, in order to retrieve the desired product $\mathbf C^{(\kappa)}=\mathbf {AB}^{(\kappa)}$.

To guarantee the \textit{secrecy of input matrix} $\mathbf A$, in a manner similar to \eqref{eq_security_constraint}, we have the constraint
\begin{equation}\label{eq_security}
I(\mathbf A_{p}^{(\kappa)},\mathbf B_{p}^{(\kappa)},\mathbf q_{p}^{(\kappa)},\mathcal{B};\mathbf A)=0,
\end{equation}
for all $p\in [1,P]$. Following the PIR formulation on \cite{kim2019private}, in order to ensure the \textit{privacy of index} $\kappa$,  for some value of $\kappa$ the information available at each worker should be statistically indistinguishable from that available for any other value  $\kappa'\neq \kappa$. Mathematically, for all $\kappa,\kappa'\in [1,L]$ with $\kappa'\neq \kappa$ and for all workers $p\in[1,P]$, we have the condition
\begin{equation}\label{eq_privacy}
(\mathbf q_p^{(\kappa)},\mathbf A_p^{(\kappa)},\mathbf C_p^{(\kappa)}, \mathcal B )\sim (\mathbf q_p^{(\kappa')},\mathbf A_p^{(\kappa')},\mathbf C_p^{(\kappa')}, \mathcal B),
\end{equation}
that is, the joint distribution of variables $(\mathbf q_p^{(\kappa')},\mathbf A_p^{(\kappa')},\mathbf C_p^{(\kappa')}, \mathcal B)$ should be the same  for any pair of index values $\kappa'\neq \kappa$.
Finally, the \textit{correct decoding} requirement is defined as in \eqref{eq_decidability_constraint}, that is
\begin{equation}\label{eq_decidability_constraintPri}
H(\mathbf{AB}^{(\kappa)}|\{\mathbf C_p^{(\kappa)}\}_{p\in \mathcal{P}_R})=0. 
\end{equation}
A coding and decoding strategy that satisfies conditions \eqref{eq_security}, \eqref{eq_privacy}, and \eqref{eq_decidability_constraintPri} is said to be \textit{feasible}. 
For given parameters $m$ and $n$ the performance is measured by the pair $(P_R,C_L)$, with $P_C=1$, where $C_L$ is the communication load defined in \eqref{CommLoad}.
 
\section{Background: Generalized PolyDot Code  without Security Constraint}\label{Sec_GPD}
 
In this section, we consider the system model shown in Fig.~\ref{FigComputSyst} and review the GPD construction first proposed in
\cite{fahim2017optimal} and later improved in \cite{yu2018straggler,dutta2018unified} for the special case of no secrecy constrains, i.e., $P_C=0$. In the process, we propose a novel intuitive interpretation of GPD encoding and decoding based on the distributed computation of samples from convolutions via $z$-transforms.
   
We start by recalling that the GPD coding scheme  achieves the best currently known trade-off between recovery threshold $P_R$ and communication load $C_L$ for $P_C=0$, i.e., under no security constraint. The entangled polynomial codes of \cite{yu2018straggler} have the same properties in terms of $(P_R,P_C)$.
 The GPD codes for $P_C=0$ also achieve the optimal recovery threshold among all linear coding strategies in the cases of $t = 1$ or $d = 1$, also they minimize the recovery threshold for the minimum communication load $C_{L,\min}$ \cite{yu2017polynomial,yu2018straggler}. 
 \begin{figure*}[t!]
 	\begin{center}
 		\includegraphics[width=10.9cm]{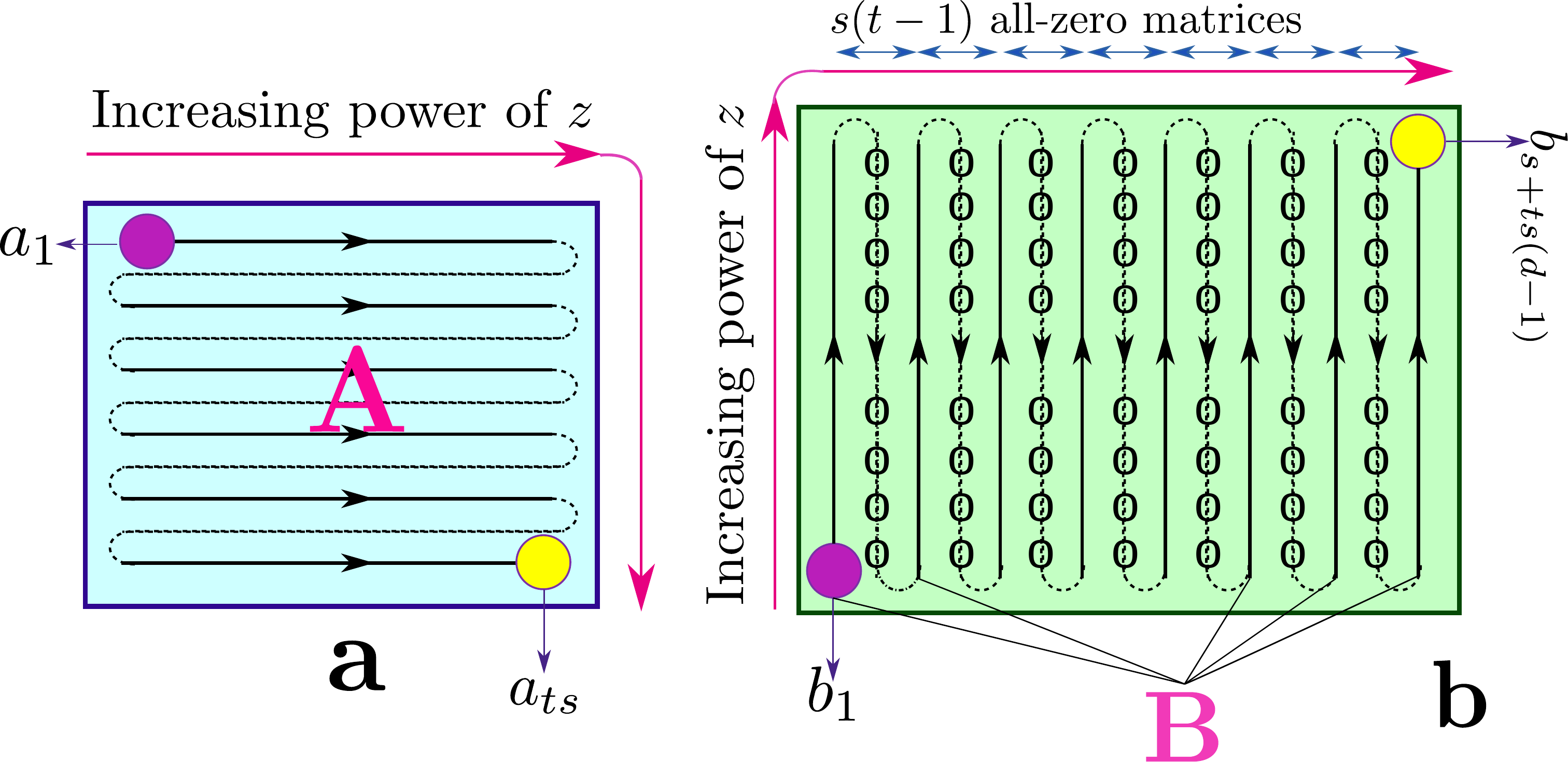}~\caption{\footnotesize 
 			Construction of the time sequences $\mathbf{a}$ and $\mathbf{b}$ used to define the generalized PolyDot (GPD) code. The zero dashed lines in  $\mathbf{b}$ indicates all-zero block sequences. Each solid arrows in $\mathbf{a}$ and $\mathbf{b}$ shows a distinct row of $\mathbf{A}$ and a column of $\mathbf{B}$, respectively.  }~\label{figGPD}
 	\end{center}
 	
 \end{figure*}
 
The GPD code splits the data matrices $\mathbf{A}$ and $\mathbf{B}$ both horizontally and vertically as
\begin{equation}\label{Polydot}
{ 
	\mathbf{A}=
	\left[ {\begin{array}{cccc}
		\mathbf{A}_{1,1}&\ldots&\mathbf{A}_{1,s}\\
		\vdots&\ddots&\vdots\\
		\mathbf{A}_{t,1}&\ldots&\mathbf{A}_{t,s}
		\end{array} } \right],}
\quad
\mathbf{B}=
\left[ {\begin{array}{cccc}
	\mathbf{B}_{1,1}&\ldots&\mathbf{B}_{1,d}\\
	\vdots&\ddots&\vdots\\
	\mathbf{B}_{s,1}&\ldots&\mathbf{B}_{s,d}
	\end{array} } \right].
\end{equation}
The parameters $s,t$, and $d$ can be set arbitrarily under the constraints $m=ts$ and $n=sd$. Note that polynomial codes set $s=1$, while MatDot codes have $t=d=1$ \cite{dutta2018optimal}.
All  sub-matrices $\mathbf{A}_{i,j}$ and $\mathbf{B}_{k,l}$ have dimensions $T/t\times S/s$ and $S/s\times D/d$, respectively.
The GPD code computes each block $(i,j)$ of the product $\mathbf{C}=\mathbf{AB}$, namely $\mathbf{C}_{i,j}=\sum_{k=1}^s\mathbf{A}_{i,k}\mathbf{B}_{k,j}$, for $i\in[1,t]$ and $j\in[1,d]$, in a distributed fashion. 
This is done by means of polynomial encoding and polynomial interpolation. As we review next, the computation of block $\mathbf{C}_{i,j}$ can be interpreted as the evaluation of the middle sample of the convolution $\mathbf{c}_{i,j}=\mathbf{a}_i*\mathbf{b}_j$ between the block sequences $\mathbf{a}_i=[\mathbf{A}_{i,1},\ldots,\mathbf{A}_{i,s}]$ and $\mathbf{b}_j=[\mathbf{B}_{s,j},\ldots,\mathbf{B}_{1,j}]$.
In fact, the $s$th sample of the block sequence $\mathbf{c}_{i,j} $ equals $\mathbf{C}_{i,j}$, i.e., $[\mathbf{c}_{i,j}]_s=\mathbf{C}_{i,j}$.
The computation is carried out distributively in the frequency domain by using $z$-transforms with different workers being assigned distinct samples in the frequency domain.

To elaborate, define the block sequence $\mathbf{a}$ obtained by concatenating the block sequences $\mathbf{a}_i$ as $\mathbf{a}=\{\mathbf{a}_1,\mathbf{a}_2,\ldots,\mathbf{a}_t\}$. 
Pictorially, a sequence $\mathbf{a}$ is obtained from the matrix $\mathbf{A}$ by reading the blocks in the left-to-right top-to-bottom order, as seen in Fig.~\ref{figGPD}.
We also introduce the longer time block sequence $\mathbf{b}$ as 
\begin{equation}\label{eqTimeSeqB}
\mathbf{b}=\{\mathbf{b}_1,\mathbf{0},\mathbf{b}_2,\mathbf{0},\ldots,\mathbf{b}_d\},
\end{equation}
with ${\mathbf{0}}$ being a block sequence of $s(t^*-1)$ all-zero block matrices with dimensions $S/s\times D/d$. The sequence $\mathbf{b}$ can be obtained from the matrix $\mathbf{B}$ by following the bottom-to-top left-to-right order shown in Fig.~\ref{figGPD} and by adding the all-zero block sequences between any two columns of the matrix  $\mathbf{B}$.  

In the frequency domain, the $z$-transforms of sequences $\mathbf{a}$ and $\mathbf{b}$ are obtained as
\begin{align}
\mathbf{F}_{\mathbf{a}}(z)=&\sum_{r=0}^{ts-1} [\mathbf{a}]_{r+1}z^{r}=\sum_{i=1}^t\sum_{j=1}^s\mathbf{A}_{i,j}z^{s(i-1)+j-1},\label{eqDuttapAz}\\
\mathbf{F}_{\mathbf{b}}(z)=&\sum_{r=0}^{s-1+ts(d-1)}[\mathbf{b}]_{r+1}z^{r}=\sum_{k=1}^s\sum_{l=1}^d\mathbf{B}_{k,l}z^{s-k+ts(l-1)},\label{eqDuttapBz}
\end{align}
respectively.
The master server evaluates the polynomials $\mathbf{F}_{\mathbf{a}}(z)$ and $\mathbf{F}_{\mathbf{b}}(z)$ in $P$ non-zero distinct points $z_1,\ldots, z_P\in \mathbb{F}$ and sends the corresponding linearly encoded matrices $\mathbf{A}_p=\mathbf{F}_{\mathbf{a}}(z_p)$ and $\mathbf{B}_p=\mathbf{F}_{\mathbf{b}}(z_p)$ to server $p$. The encoding functions are hence given by the polynomial evaluations \eqref{eqDuttapAz} and \eqref{eqDuttapBz}, for $z_1,\ldots,z_p$. Server $p$ computes the multiplication $\mathbf{F}_{\mathbf{a}}(z_p)\mathbf{F}_{\mathbf{b}}(z_p)$ and sends it to the master server.
The master server computes the inverse $z$-transform for the received products $\{\mathbf{A}_p\mathbf{B}_p\}_{p\in \mathcal{P}_R} =\{\mathbf{F}_{\mathbf{a}}(z_p)\mathbf{F}_{\mathbf{b}}(z_p)\}_{p\in \mathcal{P}_R}$, obtaining the convolution $\mathbf{a}*\mathbf{b}$.

From the convolution $\mathbf{a}*\mathbf{b}$ we can see that the master server is able to compute all the desired blocks $\mathbf{C}_{i,j}$ by reading the middle samples of the convolutions $\mathbf{c}_{i,j}=\mathbf{a}_i*\mathbf{b}_j$ from samples of the sequence $\mathbf{c}=\mathbf{a}*\mathbf{b}$ in the order $[\mathbf{c}]_{s-1}=\mathbf{C}_{1,1},[\mathbf{c}]_{2s-1}=\mathbf{C}_{2,1},\ldots,[\mathbf{c}]_{ts-1}=\mathbf{C}_{t,1},[\mathbf{c}]_{s-1+t^*s}=\mathbf{C}_{1,2},\ldots,[\mathbf{c}]_{ts-1+t^*s}=\mathbf{C}_{t,2},\ldots$. Note that, in particular, the zero block subsequences added to sequence $\mathbf{b}$ ensure that no interference from the other convolutions, $\mathbf{c}_{i',j'}$ affects the middle ($s$th) sample of a convolution $\mathbf{c}_{i,j}$ with $i'\neq i$ and $j'\neq j$.

To carry out the inverse transform, the master server needs to collect as many values $\mathbf{F}_{\mathbf{a}}(z_p)\mathbf{F}_{\mathbf{b}}(z_p)$ as there are samples of the sequence $\mathbf{a}*\mathbf{b}$, yielding the recovery threshold
\begin{equation}\label{eq_RT}
P_R=tsd+s-1.
\end{equation}
Equivalently, in terms of the underlying polynomial interpretation, the master server needs to collect a number of evaluations of the polynomial $\mathbf{F}_{\mathbf{a}}(z)\mathbf{F}_{\mathbf{b}}(z)$ equal to the degree of $\mathbf{F}_a(z)\mathbf{F}_b(z)$ plus one.
This computation is of complexity order $\mathcal{O}(TDP_R(\log(P_R))^2)$  \cite{dutta2018optimal}.
 Furthermore, the communication load is given as
\begin{equation}\label{eq_CL}
C_L=P_R\frac{TD}{td},
\end{equation}
where $TD/(td)$ is the size of each matrix $\mathbf{F}_{\mathbf{a}}(z)\mathbf{F}_{\mathbf{b}}(z)$.

\section{Secure PolyDot Code}\label{sec_SecPDC}
In this section, we propose a novel extension of the GPD code that is able to ensure the secrecy constraint for any $P_C<P$. We also derive the corresponding achievable set of triples $(P_C,P_R,C_L)$. As we will discuss, the projection of this set onto the plane defined by the condition $P_C=0$ includes the set of pairs $(P_R,C_L)$ in \eqref{eq_RT} and \eqref{eq_CL} obtained by the GPD code \cite{dutta2018unified}.
The proposed secure GPD (SGPD) code augments matrices $\mathbf{A}$ and $\mathbf{B}$ by adding $P_C$ random block matrices to the input matrices $\mathbf{A}$ and $\mathbf{B}$, in a manner similar to prior works  \cite{nodehi2018limited,yu2018lagrange,chang2018capacity,kakar2019capacity,rafael2018codes}, 
yielding augmented matrices $\mathbf{A}^*$ and $\mathbf{B}^*$. As we will see, a direct application of the GPD codes to these matrices is suboptimal. 

   \begin{figure*}[t!]
	\begin{center}
		\includegraphics[width=11.6cm]{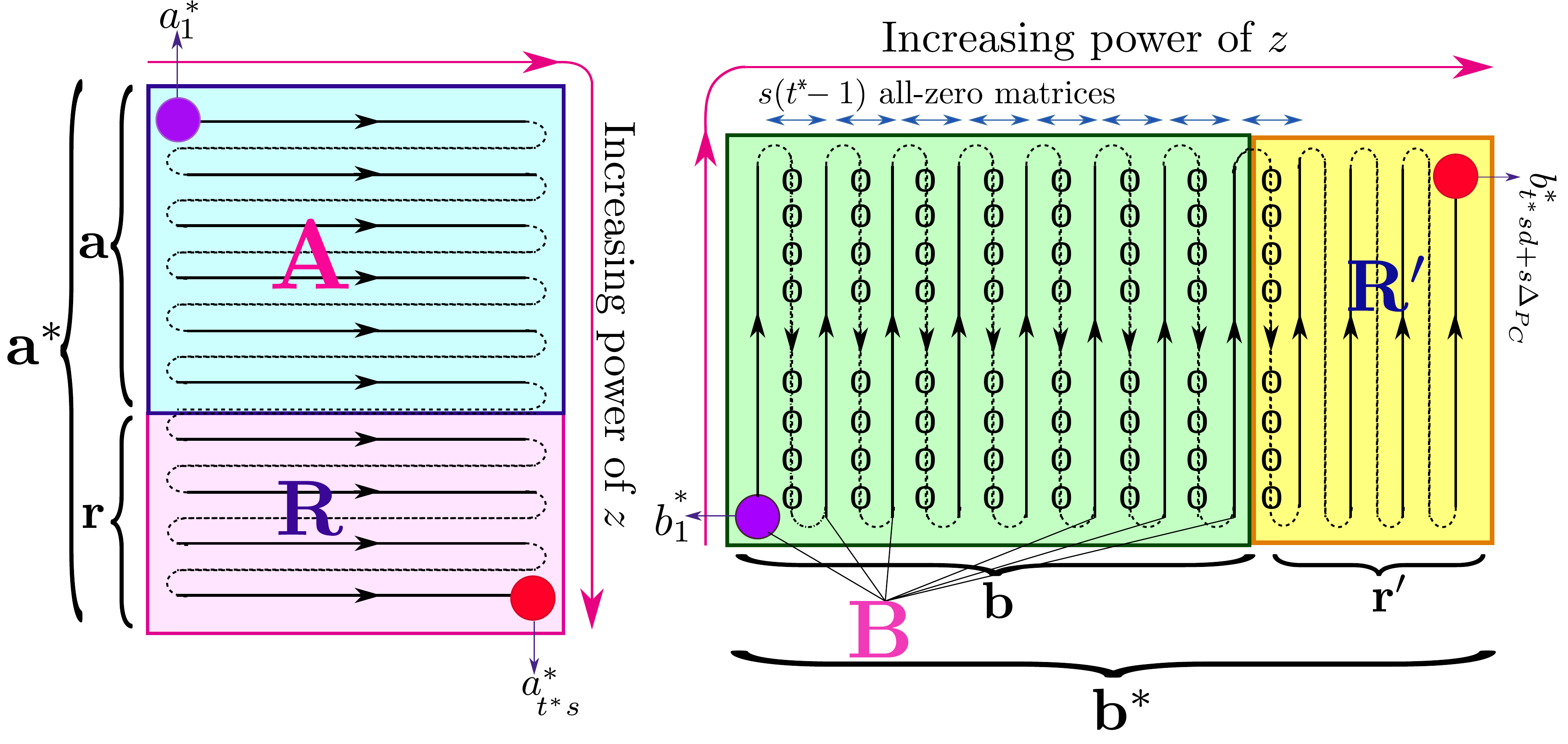}~\caption{\footnotesize
			Construction of the time block sequences
			$\mathbf{a}^*=[\mathbf{a},\mathbf{r}]$ and
			$\mathbf{b}^*=[\mathbf{b},\mathbf{r}']$ in
			\eqref{eqTimeSeqAR} and \eqref{eqTimeSeqBR} used to define
			the SGPD code for  the case $s< t$. The zero dashed lines in $\mathbf{b}$  and $\mathbf{r}'$ indicate all-zero block sequences.  }~\label{SecGPD} 
	\end{center}
	\vspace{-5ex}
\end{figure*}
In contrast, we propose a novel way to construct sequences $\mathbf{a}^*$ and $\mathbf{b}^*$ from matrices $\mathbf{A}^*$ and $\mathbf{B}^*$ that enables the definition of a more efficient code by means of the $z$-transform approach discussed in the previous section. To this end, we follow the design criterion of decreasing the recovery threshold $P_R$ for a given communication load $C_L$. Based on the discussion in the previous section, this goal can be realized by decreasing the length of the sequence $\mathbf{c}^*=\mathbf{a}^**\mathbf{b}^*$, which can in turn be ensured by reducing the length of the sequence $\mathbf{b}^*$ for a given length of the sequence $\mathbf{a}^*$. We accomplish this objective by {$\emph (i)$} adaptively appending rows \textit{or} columns with random elements to matrix $\mathbf{A}$, and, correspondingly columns \textit{or} rows to $\mathbf{B}$, which can reduce the recovery threshold; and {$\emph (ii)$} modifying the zero padding procedure (see Fig.~\ref{figGPD}) for the construction of sequence $\mathbf{b}^*$. In order to account for point {\emph{(i)}}, we consider separately the two cases $s< t$ and $s\geq t$.

\subsection{Secure Generalized PolyDot Code: The $s< t$ Case}\label{subsec1}
As illustrated in Fig.~\ref{SecGPD}, when $s< t$, we augment the input matrices $\mathbf{A}$ and $\mathbf{B}$ by adding 
\begin{equation}\label{eq_delta1}
\Delta_{P_C}\overset{\Delta}{=} \left\lceil 
\frac{P_C}{s} \right\rceil,
\end{equation}
random row and column blocks to matrices $\mathbf{A}$ and $\mathbf{B}$, respectively. Accordingly, the $t^*\times s$ augmented block matrix $\mathbf{A}^*$ with $t^*=t+\Delta_{P_C}$ is obtained as
  \begin{equation}\label{eq_matrixA}
   { 
   	\mathbf{A}^*=\left[ {\begin{array}{c}
   	   		\mathbf{A}\\\mathbf{R} 
   	   		\end{array} } \right]=
   	\left[ {\begin{array}{cccc}
   		\mathbf{A}_{1,1}&\ldots&\mathbf{A}_{1,s}\\
   		\vdots&\ddots&\vdots\\
   		\mathbf{A}_{t,1}&\ldots&\mathbf{A}_{t,s}\\
   		\mathbf{R}_{1,1}&\ldots&\mathbf{R}_{1,s}\\
   		\vdots&\ddots&\vdots\\
   		\mathbf{R}_{\Delta_{{P_C},1}}&\ldots&\mathbf{R}_{\Delta_{{P_C},s}}
   		\end{array} } \right],}
  \end{equation}
  while the $s\times d^*$ augmented matrix $\mathbf{B}^*=[\mathbf{B}~\mathbf{R}']$ with $d^*=d+\Delta_{P_C}$ is obtained as   \begin{equation}\label{eq_matrixB}
  \mathbf{B}^*=
  \left[ {\begin{array}{cccccc}
  	\mathbf{B}_{1,1}&\ldots&\mathbf{B}_{1,d}&\mathbf{R}'_{s,1}&\ldots&\mathbf{R}'_{s,\Delta_{P_C}}\\
  	\vdots&\ddots&\vdots&\vdots&\ddots&\vdots\\
  	\mathbf{B}_{s,1}&\ldots&\mathbf{B}_{s,d}&\mathbf{R}'_{1,1}&\ldots&\mathbf{R}'_{1,\Delta_{P_C}}
  	\end{array} } \right].
  \end{equation} 
  In \eqref{eq_matrixA} and \eqref{eq_matrixB}, if $s$ divides $P_C$, all block matrices $\mathbf{R}_{i,j}\in \mathbb{F}^{\frac{T}{t}\times\frac{S}{s}}$ and $\mathbf{R}'_{i,j}\in \mathbb{F}^{\frac{S}{s}\times\frac{D}{d}}$ are generated with i.i.d. uniform random elements in $\mathbb{F}$. Otherwise, if $\Delta_{P_C}-P_C/s >0$, the last $s\Delta_{P_C}-P_C$ matrices  in \eqref{eq_matrixA}, with right-to-left ordering in the last row of $\mathbf{R}_{i,j}$, and in \eqref{eq_matrixB} with top-to-bottom ordering in the last column of  $\mathbf{R}'_{i,j}$, are all-zero block matrices.

As illustrated in Fig.~\ref{SecGPD}, in the SGPD scheme, the block sequence $\mathbf{a}^*$ is defined in the same way as in the conventional GPD, yielding  
  \begin{equation}\label{eqTimeSeqAR}
  \mathbf{a}^*=\{\mathbf{a}_1,\ldots,\mathbf{a}_t,\mathbf{r}_1,\ldots,\mathbf{r}_{\Delta_{P_C}}\},
  \end{equation}
where $\mathbf{r}_i$ is the $i$th row of the block matrix $\mathbf{R}$, $i\in[1,\Delta_{P_C}]$. We also define the time block sequence $\mathbf{b}^*=\{\mathbf{b},\mathbf{r}'\}$ as
 \begin{equation}\label{eqTimeSeqBR}
  \mathbf{b}^*=\{\mathbf{b}_1,\mathbf{0},\mathbf{b}_2,\mathbf{0},\ldots,\mathbf{b}_d,\mathbf{0},\mathbf{r}'_1,\mathbf{r}'_2,\ldots,\mathbf{r}'_{\Delta_{P_C}}\},
  \end{equation}
   where $\mathbf{0}$ is block sequences of $s(t^*-1)$ all-zero block matrices, respectively, with dimensions $S/s\times D/d$, while $\mathbf{r}'_j$ is the $j$th column of the random matrix $\mathbf{R}'$. The key novel idea of this construction is that no zero matrices are introduced between the columns of matrix $\mathbf{R}'$.
   As shown in Theorem \ref{SecurThm} below, this construction allows the master server to recover all the desired submatrices $\mathbf{C}_{i,j}$ for $i\in[1,t]$ and $j\in[1,d]$ from the middle samples of the convolutions $\mathbf{c}_{i,j}=\mathbf{a}_i*\mathbf{b}_j$ (see Fig.~\ref{figconv} for an illustration).
\begin{Theorem}\label{SecurThm}
For a given security level $P_C<P$, the proposed SGPD code achieves the recovery threshold $P_R$
\begin{equation}\label{RT_1}
{\small
\begin{cases}
tsd+s-1,&\text{ if } P_C=0,\\
t^*s(d+1)+s\Delta_{P_C}-1,& \text{ if } P_C\geq 1 \text{ and }\Delta_{P_C}=\frac{P_C}{s},\\
t^*s(d+1)-s\Delta_{P_C}+2P_C-1,&  \text{ if } P_C\geq 1 \text{ and } \Delta_{P_C}>\frac{P_C}{s},
\end{cases}
}
\end{equation}
and the communication load \eqref{eq_CL},
where  $t^*=t+\Delta_{P_C}$ and $d^*=d+\Delta_{P_C}$ for any integer values $t,s$, and $d$ such that $s< t$, $m=ts$, and $n=sd$.
\end{Theorem}
 \begin{proof}
The $z$-transform of sequences $\mathbf{a}^*$ and $\mathbf{b}^*$ are given respectively as
\begin{align}
\mathbf{F}_{\mathbf{a}^*}(z)&
= \underbrace{\sum_{i=1}^{t}\sum_{j=1}^{s} \mathbf{A}^*_{i,j}z^{s(i-1)+(j-1)}}_{\overset{\Delta}{=}~\mathbf F_1(z)}\nonumber\\
&+\underbrace{\sum_{i=t+1}^{t^*}\sum_{j=1}^{s} \mathbf{A}^*_{i,j} z^{s(i-1)+j-1}}_{\overset{\Delta}{=}~\mathbf F_2(z)},\label{eqsec1_pAM1li}\\
\mathbf{F}_{\mathbf{b}^*}(z)&= 	 	 \underbrace{\sum_{k=1}^{s}\sum_{l=1}^{d}\mathbf{B}^*_{k,l}z^{s-k+t^*s(l-1)}}_{\overset{\Delta}{=} ~\mathbf F_3(z)}\nonumber\\
&+\underbrace{\sum_{k=1}^{s}\sum_{l=d+1}^{d^*}\mathbf{B}^*_{k,l} z^{t^*sd+s(l-d)-k }}_{\overset{\Delta}{=}~\mathbf F_4(z)}.\label{eqsec1_pBM1li}
\end{align}
\begin{figure}[t!]
	\begin{center}
		\includegraphics[scale=0.77]{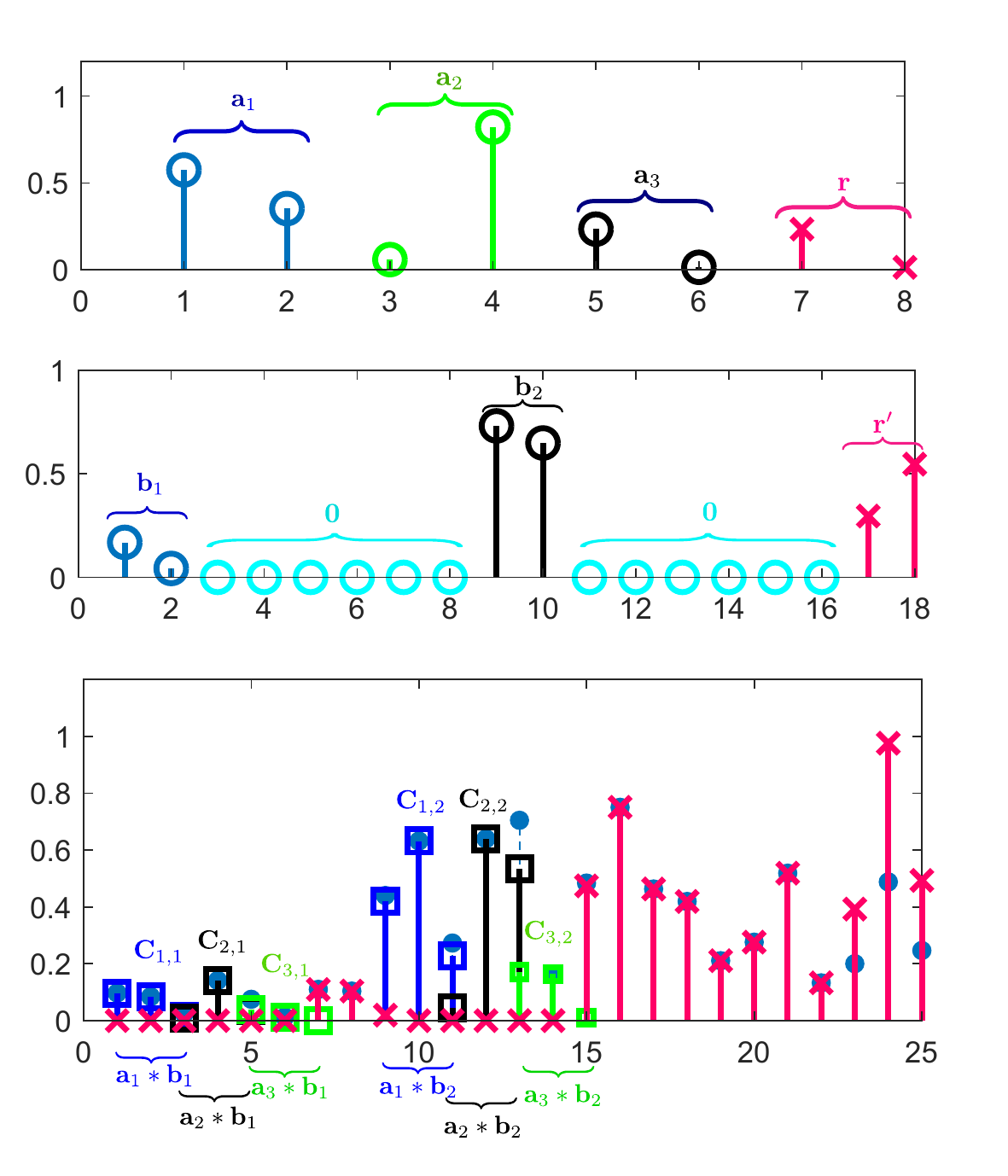}\vspace{-.1cm}~\caption{\footnotesize{Outcome of the communication $\mathbf C_{i,j}=\mathbf a_i*\mathbf b_j$ for $t=3,s=2,d=2$, and 		$P_C=2$. Dashed blue stems with filled markers represent 			the convolution $\mathbf{c}^*$. Individual convolutions $\mathbf{c}_{i,j}$ are shown in different colors with square markers. Contributions from one or both random matrices are shown as red crosses. The desired submatrices $\mathbf{C}_{i,j}$ are seen to equal the	corresponding samples from the sequence $\mathbf{c}^*$, associated with the center points of the individual  convolutions. }}~\label{figconv}
	\end{center}
	\vspace{-5ex}
\end{figure}
The master server evaluates $\mathbf{F}_{\mathbf{a}^*}(z)$ and $\mathbf{F}_{\mathbf{b}^*}(z)$ at $P$ non-zero distinct points $z_1,\ldots,z_P\in \mathbb{F}$, which define the encoding functions, and sends both matrices $\mathbf{A}_p=\mathbf{F}_{\mathbf{a}^*}(z_p)$ and $\mathbf{B}_p=\mathbf{F}_{\mathbf{b}^*}(z_p)$ to worker $p$. Worker $p$ performs the multiplication $\mathbf{F}_{\mathbf{a}^*}(z_p)\mathbf{F}_{\mathbf{b}^*}(z_p)$, and sends the results back to the master server. 
To reconstruct all blocks $\mathbf{C}_{i,j}$ of matrix $\mathbf{C}=\mathbf{AB}$, the master server carries out a polynomial interpolation, or equivalently, it computes the inverse $z$-transform, upon  receiving a number of multiplication results equal to at least the length of the sequence $\mathbf{c}^*=\mathbf{a}^**\mathbf{b}^*$. 
As we detail next, the $(i,l)$ block $\mathbf{C}_{i,l}=\sum_{r=1}^{s}\mathbf{A}_{i,r}\mathbf{B}_{r,l}$, for all $i\in[1,t]$ and $l\in[1,d]$, of matrix $\mathbf{C}=\mathbf{AB}$ can be seen equal to the $(si-1+(l-1)t^*s)$th sample of the convolution $\mathbf{c}^*=\mathbf{a}^**\mathbf{b}^*$. An illustration can be found in Fig.~\ref{figconv}.

To see this, we first note that, by the properties of GPD codes, matrix $\mathbf{C}_{i,l}$ is the coefficient of the monomial $z^{si-1+(l-1)t^*s}$ in $\mathbf{F}_{1}(z)\mathbf{F}_{3}(z)$. Note that this holds since the polynomial $\mathbf{F}_1(z)$ and $\mathbf{F}_3(z)$ are defined as GPD codes. We now need to show that no other contribution to this term arises from the products $\mathbf{F}_1(z)\mathbf{F}_4(z)$, $\mathbf{F}_2(z)\mathbf{F}_3(z)$, and $\mathbf{F}_2(z)\mathbf{F}_4(z)$. 
 The terms in the product $\mathbf{F}_1(z)\mathbf{F}_4(z)$ have exponents $(t^*sd+s(i-1)+s(l-d)-1)$, for $i\in[1,t]$ and $l\in[d+1,d^*]$, which do not include the desired values $(si-1+(l-1)t^*s)$ for $i\in[1,t]$ and $l\in[1,d]$.
 A similar discussion applies to the product $\mathbf{F}_2(z)\mathbf{F}_3(z)$, whose exponents are $(s(i+t^*l-t^*)-1)$, for $i\in[t+1,t^*]$ and $l\in[1,d]$, and $\mathbf{F}_2(z)\mathbf{F}_4(z)$, whose exponents are $(t^*sd+s(i-1)+s(l-d)-1)$, for $i\in[t+1,t^*]$ and $l\in[d+1,d^*]$.
 
 \begin{figure*}[t!]
 	\begin{center}
 		\includegraphics[width=12cm]{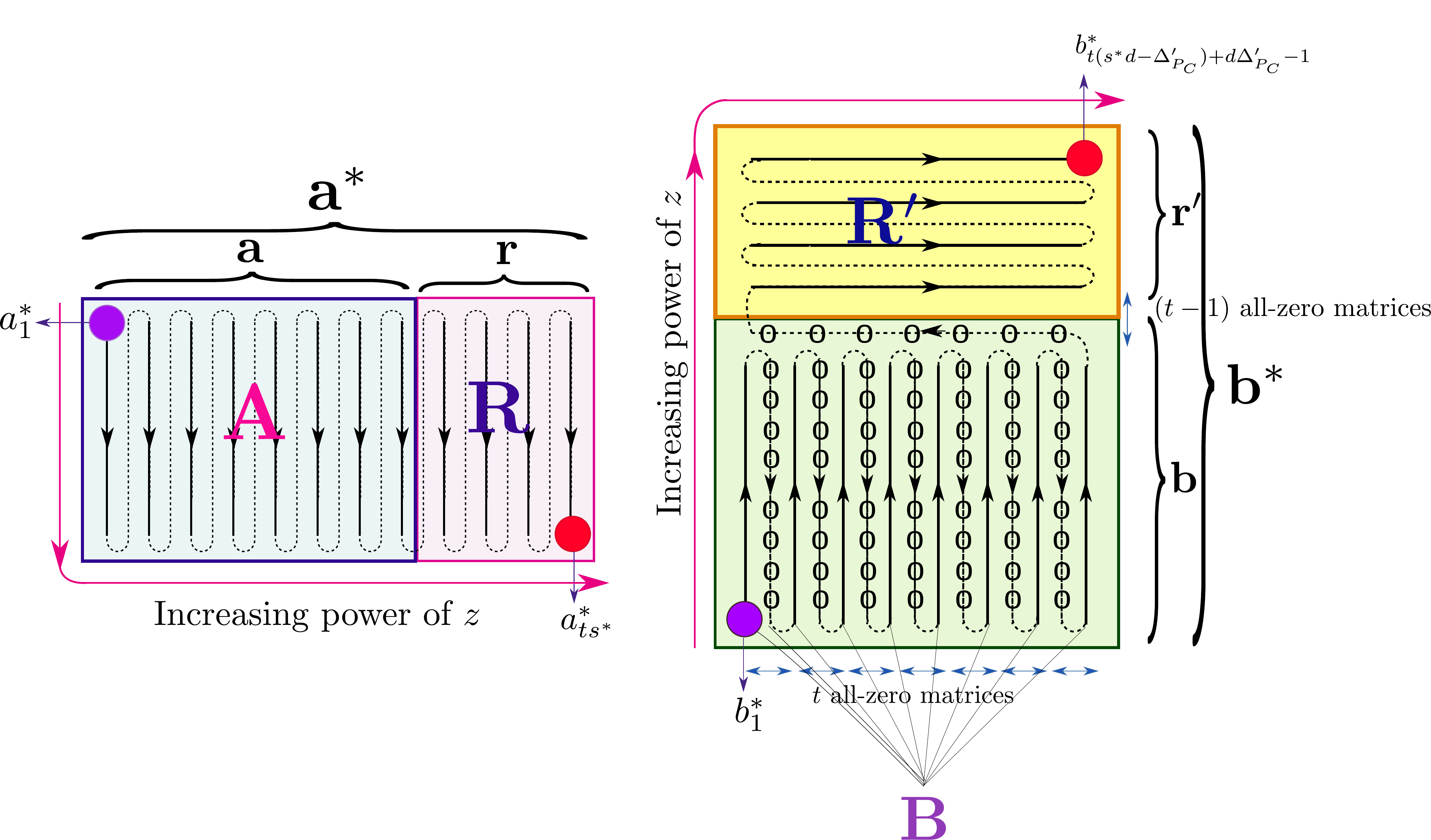}~\caption{\footnotesize
 			Construction of the time block sequences $\mathbf{a}^*$
 			and $\mathbf{b}^*$ in \eqref{eqTimeSeqAR2} and
 			\eqref{eqTimeSeqBB} used to define the secure generalized PolyDot (SGPD) code for  the case $s\geq t$. The solid line and the zero dashed lines in $\mathbf{b}^*$ indicate  columns of $\mathbf{B}$ and all-zero block sequences, respectively.   }~\label{SecGPD2}
 	\end{center}
 	\vspace{-5ex}
 \end{figure*}
 In order to recover the convolution $\mathbf{c}^*$, the master server needs to collect a number of values of the product $\mathbf{F}_\mathbf{a}(z)\mathbf{F}_{\mathbf{b}}(z)$ equal to the length of the sequence $\mathbf{c}^*$, which can be computed as the degree $\deg \left(\mathbf{F}_\mathbf{a}(z)\mathbf{F}_{\mathbf{b}}(z)\right)+1$, where $\deg (\mathbf{F}_\mathbf{a}(z)\mathbf{F}_{\mathbf{b}}(z))$ is
 \begin{equation}\label{eq_deg}
 \begin{cases}
 t^*s(d+1)+s\Delta_{P_C}-1,& \text{ if }\Delta_{P_C}=\frac{P_C}{s},\\
 dst^*-s\Delta_{P_C}+2P_C+t-2,& \text{ if }\Delta_{P_C}>\frac{P_C}{s}.
\end{cases}
 \end{equation}
For $P_C\geq 1$ this implies the recovery threshold $P_R$ in \eqref{RT_1}. 
The communication load $C_L$ in \eqref{eq_CL} follows from the fact that there are $TD/(td)$ entries in $\mathbf{F}_{\mathbf{a}^*}(z_p)\mathbf{F}_{\mathbf{b}^*}(z_p)$, for all $p\in [1,P_R]$.  
 
The security constraint \eqref{eq_security_constraint} can be proved in a manner similar to \cite{chang2018capacity} by the following steps:
\begin{align}\label{eq_sec}
 &I(\mathbf{A},\mathbf{B};\mathbf A_{\mathcal{P}},\mathbf B_{\mathcal{P}})\nonumber\\
 =&H(\mathbf A_{\mathcal{P}},\mathbf B_{\mathcal{P}})-H(\mathbf A_{\mathcal{P}},\mathbf B_{\mathcal{P}}|\mathbf{A},\mathbf{B})\nonumber\\
 \overset{(a)}{=}&H(\mathbf A_{\mathcal{P}},\mathbf B_{\mathcal{P}})-H(\mathbf A_{\mathcal{P}},\mathbf B_{\mathcal{P}}|\mathbf{A},\mathbf{B})\nonumber\\
 &+H(\mathbf A_{\mathcal{P}},\mathbf B_{\mathcal{P}}|\mathbf{A},\mathbf{B},\mathbf{R}_1,\ldots,\mathbf{R}_{P_C},\mathbf{R}'_1,\ldots,\mathbf{R}'_{P_C})\nonumber\\
 \overset{}{=}&H(\!\mathbf A_{\mathcal{P}},\mathbf B_{\mathcal{P}}\!)-I(\!\mathbf A_{\mathcal{P}},\mathbf B_{\mathcal{P}};\mathbf{R}_1,\ldots,\mathbf{R}_{P_C},\mathbf{R}'_1,\ldots,\mathbf{R}'_{P_C}|\mathbf{A},\mathbf{B})\nonumber\\
 \overset{}{=}&H(\mathbf A_{\mathcal{P}},\mathbf B_{\mathcal{P}})-H(\mathbf{R}_1,\ldots,\mathbf{R}_{P_C},\mathbf{R}'_1,\ldots,\mathbf{R}'_{P_C}|\mathbf{A},\mathbf{B})\nonumber\\&+H(\mathbf{R}_1,\ldots,\mathbf{R}_{P_C},\mathbf{R}'_1,\ldots,\mathbf{R}'_{P_C}|\mathbf{A},\mathbf{B},\mathbf A_{\mathcal{P}},\mathbf B_{\mathcal{P}})\nonumber\\
 \overset{(b)}{=}& 
 H(\mathbf A_{\mathcal{P}},\mathbf B_{\mathcal{P}})-H(\mathbf R_{1},\ldots,\mathbf{R}_{P_C},\mathbf R'_{1},\ldots,\mathbf{R}'_{P_C})\nonumber\\
 \overset{(c)}{\leq}&H(\mathbf A_{\mathcal{P}})+H(\mathbf B_{\mathcal{P}})-\sum_{p=1}^{P_C}H(\mathbf{R}_p)-\sum_{p=1}^{P_C}H(\mathbf{R}'_p)\nonumber\\
 \overset{(d)}{=}&H(\mathbf A_{\mathcal{P}})+H(\mathbf B_{\mathcal{P}})-P_C\frac{TS}{m}\log|\mathbb{F}|-P_C\frac{SD}{n}\log|\mathbb{F}|\nonumber\\
 \overset{(e)}{\leq}&\sum_{p=1}^{P_C}H(\mathbf{A}_p)+\sum_{p=1}^{P_C}H(\mathbf{B}_p)-P_C\frac{TS}{m}\log|\mathbb{F}|-P_C\frac{SD}{n}\log|\mathbb{F}|\nonumber\\
 \overset{(f)}{=}&P_C\frac{TS}{m}\log|\mathbb{F}|+P_C\frac{SD}{n}\log|\mathbb{F}|-P_C\frac{TS}{m}\log|\mathbb{F}|\nonumber\\
 &-P_C\frac{SD}{n}\log|\mathbb{F}|\nonumber\\
 =&0,
\end{align}
where $(a)$ follows from the definition of encoding functions, since $\mathbf{A}_{\mathcal{P}}$ is a deterministic function of $\mathbf{A}$ and $\mathbf{R}_p$, and $\mathbf{B}_{\mathcal{P}}$ is a deterministic function of $\mathbf{B}$ and $\mathbf{R}'_p$, respectively, for all $p\in[1,P_C]$; 
$(b)$ follows from \eqref{eqsec1_pAM1li} and \eqref{eqsec1_pBM1li}, since from $P_R$ polynomial evaluations $\mathbf{A}_{\mathcal{P}}$ and $\mathbf{B}_{\mathcal{P}}$ in \eqref{eqsec1_pAM1li} and \eqref{eqsec1_pBM1li} we can recover $2P_C$ unknowns when the coefficients $\mathbf{A}_{i,j}$ and $\mathbf{B}_{k,l}$ are known, given that we have $P_R\geq 2P_C$; $(c)$ and $(d)$ follows since $\mathbf{R}_p$ and $\mathbf{R}'_p$ are independent uniformly distributed entries; 
$(e)$ follows by upper bounding the joint entropy using the sum of individual entropies; and $(f)$ follows from an argument similar to $(d)$. Hence, the proposed scheme is information-theoretically secure. 
\end{proof}
 
\begin{Remark}
When $P_C\geq 1$ a direct application of the GPD construction in Fig.~\ref{figGPD} would yield the larger recovery threshold
\begin{equation}
P_R = 
\begin{cases}
t^*sd^*+s-1,& \text{ if }\Delta_{P_C}=\frac{P_C}{s},\\
dst^*+s-1-2(s\Delta_{P_C}-P_C),&  \text{ if } \Delta_{P_C}>\frac{P_C}{s}.
\end{cases}
\end{equation}
\end{Remark}

\subsection{Secure Generalized PolyDot Code: The $s\geq t$ Case}\label{subsec2}
As illustrated in Fig.~\ref{SecGPD2}, when $s\geq t$, we instead augment input matrices $\mathbf{A}$ and $\mathbf{B}$ by adding 
\begin{equation}\label{delta2}
\Delta'_{P_C}\overset{\Delta}{=} \left\lceil \frac{P_C}{\min{\{ t, d \}}} \right\rceil
\end{equation}
column and row blocks to matrices $\mathbf{A}$ and $\mathbf{B}$. This can be seen to yield a smaller recovery threshold. Accordingly, the $t\times s^*$ augmented block matrix $\mathbf{A}^*=[\mathbf{A}~\mathbf{R}]$ with $s^*=s+\Delta'_{P_C}$ is obtained as 
\begin{equation}\label{eq_matrixA1}
{ 
	\mathbf{A}^*=
	\left[ {\begin{array}{cccccc}
		\mathbf{A}_{1,1}&\ldots&\mathbf{A}_{1,s}&\mathbf{R}_{1,1}&\ldots&\mathbf{R}_{1,\Delta'_{P_C}}\\
		\vdots&\ddots&\vdots&\vdots&\ddots&\vdots\\
		\mathbf{A}_{t,1}&\ldots&\mathbf{A}_{t,s}&\mathbf{R}_{t,1}&\ldots&\mathbf{R}_{t,\Delta'_{P_C}}
		\end{array} } \right] ,}
\end{equation}
while the $s^*\times d$ augmented block matrix $\mathbf{B}^*$ is defined as
\begin{equation}\label{eq_matrixB1}
\mathbf{B}^*=\left[ {\begin{array}{c}
	\mathbf{R}' \\\mathbf{B}
	\end{array} } \right]=
\left[ {\begin{array}{ccc}
	\mathbf{R}'_{\Delta'_{{P_C},1}}&\ldots&\mathbf{R}'_{\Delta'_{{P_C},d}}\\
	\vdots&\ddots&\vdots\\
	\mathbf{R}'_{1,1}&\ldots&\mathbf{R}'_{1,d}\\
	\mathbf{B}_{1,1}&\ldots&\mathbf{B}_{1,d}\\
	\vdots&\ddots&\vdots\\
	\mathbf{B}_{s,1}&\ldots&\mathbf{B}_{s,d}
	\end{array} } \right].
\end{equation}
As for \eqref{eq_matrixA1} and \eqref{eq_matrixB1}, if
$\Delta'_{P_C}-P_C/\min\{t,d\}>0$, the last $s\Delta'_{P_C}-P_C$ block
matrices in \eqref{eq_matrixA1}, with  bottom-to-top right-to-left ordering
in $\mathbf{R}$, and in \eqref{eq_matrixB1} with right-to-left top-to-bottom
ordering in $\mathbf{R}'$, are all-zero block matrices. The construction of
sequences $\mathbf{a}^*$ and $\mathbf{b}^*$ is analogous to the GPD in the non-secure case. In particular, as seen in Fig.~\ref{SecGPD2}, the time block sequence $\mathbf{a}^*$ is  
 \begin{equation}\label{eqTimeSeqAR2}
 \mathbf{a}^*=\{\mathbf{a}_1,\mathbf{r}_1,\mathbf{a}_2,\mathbf{r}_2,\ldots,\mathbf{a}_t,\mathbf{r}_t \},
 \end{equation}
whereas the block sequence $\mathbf{b}^*$ is defined as
 \begin{equation}\label{eqTimeSeqBB}
  \mathbf{b}^*=\{\mathbf{b}_1,\mathbf{0},\mathbf{b}_2,\ldots,\mathbf{0},\mathbf{b}_d,\hat{\mathbf{0}},\mathbf{r}'_{\Delta'_{P_C}}, \ldots,\mathbf{r}'_1 \}.
  \end{equation}
Here, $\mathbf{0}$ and $\hat{\mathbf{0}}$ are a block sequence of $t$ and $t-1$ all-zero block matrices with dimensions $S/s\times D/d$, respectively, while $\mathbf r'_i$ is the $i$th row of the random matrix $\mathbf R'$.
\begin{Theorem}\label{SecurThm2}
For a given security level $P_C<P$, the proposed SGPD code achieves the recovery threshold
\begin{equation}\label{eq_RT2}
P_R=t(s^*d-\Delta'_{P_C})+ts+2P_C-1
\end{equation}
and the communication load \eqref{eq_CL},
where $s^*=s+\Delta'_{P_C}$ for any integer values $t,s,$ and $d$ such that $s\geq t$, $m=ts$, and $n=sd$.
\end{Theorem}
\begin{proof}
We define the $z$-transform of sequences $\mathbf{a}^*$ and $\mathbf{b}^*$ respectively as
\begin{align}
\mathbf{F}_{\mathbf{a}^*}(z)
&=\sum_{i=1}^{t}\sum_{j=1}^{s} \mathbf{A}^*_{i,j} z^{i-1+t(j-1)}\nonumber\\
&+\sum_{i=1}^{t}\sum_{j=s+1}^{s^*} \mathbf{A}^*_{i,j} z^{i-1+t(j-1)},\label{eqsec2_pAM1li}\\
\mathbf{F}_{\mathbf{b}^*}(z)
&=\sum_{k=1+\Delta_{P_C}'}^{s^*}\sum_{l=1}^{d}\mathbf{B}^*_{k,l}z^{(s^*-k)t+ts^*(l-1)}\nonumber\\
&+\sum_{k=1}^{\Delta_{P_C}'}\sum_{l=1}^{d}\mathbf{B}^*_{k,l} z^{t(s^*d-\Delta_{P_C}')+d(\Delta_{P_C}'-k)+l-1} \label{eqsec2_pBM1li}.
\end{align}  
The $(i,l)$ block $\mathbf{C}_{i,l}=\sum_{r=1}^{s}\mathbf{A}_{i,r}\mathbf{B}_{r,l}$, for all $i\in[1,t]$ and $l\in[1,d]$, of matrix $\mathbf{C}=\mathbf{AB}$ can be seen equal to the $(i-1+t(s^*l-1))$th sample of the convolution $\mathbf{c}^*=\mathbf{a}^**\mathbf{b}^*$.
The rest of the proof follows in a manner akin to Theorem \ref{SecurThm}.
\end{proof}
\begin{Remark}\label{remark2}
The computational complexity of SGPD codes for both workers and master server can be summarized as follows.
Each worker is assigned to compute the multiplication $\mathbf C_p=\mathbf A_p\mathbf B_p$, requiring $TSD/(tsd)$ multiplications. For the master server, encoding matrices $\mathbf A_p$ and $\mathbf B_p$ at each worker amounts to evaluating $z$-transforms $\mathbf F_{\mathbf a^*}(z)$ and $\mathbf F_{\mathbf b^*}(z)$ at a random point $z_p$. This requires multiplying $z_p$ by $(ts+P_C)$ and $(sd+P_C)$ submatrices, each of dimension $T/t\times S/s$ and $S/s\times D/d$, respectively. This requires $P_C(TS/(ts)+SD/(sd))+TS+SD$ multiplications. Overall, the master server needs to carry out $ PP_C(TS/(ts)+SD/(sd))+P(TS+SD)$ multiplications.
For decoding, the master server interpolates a polynomial degree $P_R-1$ for each element in $\mathbf C$. Using a polynomial interpolation algorithm, the decoding complexity amounts to $(P_R-1)(\log(P_R-1))^2TD/(td)$ multiplications \cite{kung2009fast}.  
\end{Remark}
\begin{Example}
We now provide some numerical results of the proposed SGPD scheme. We set $P=3000$ workers and parameters $m=n=36$. The trade-off between communication load $C_L$ and recovery threshold $P_R$ for both non-secure conventional GPD codes $(P_C=0)$ and proposed SGPD code with colluding workers $P_C=11$ and $P_C=29$ is illustrated in Fig.~\ref{figSecRT_CO}. 
The figure quantifies the loss in terms of achievable pairs $(P_R,C_L)$ that is caused by the security constraint. 
\end{Example}
\begin{figure}[t!]
	\begin{center}\vspace{-1.3ex}
		\includegraphics[scale=0.6]{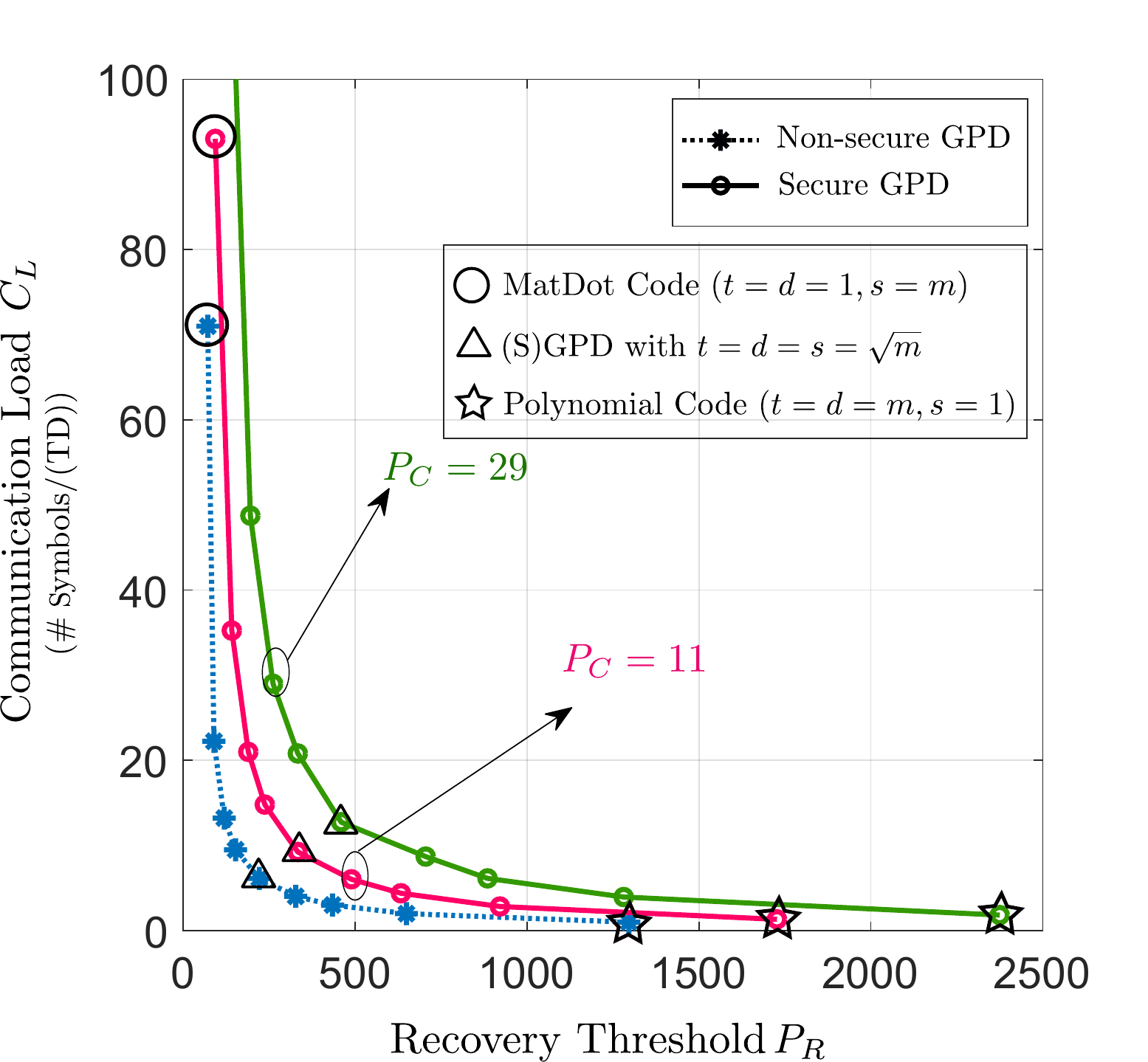}\vspace{-.1cm}~\caption{\footnotesize{Communication
				load $C_L$ versus recovery threshold $P_R$ for both
				non-secure generalized PolyDot (GPD) and secure
				generalized PolyDot (SGPD) codes ($m=n=36$ and $P=3000$ workers). 
		}}~\label{figSecRT_CO}
	\end{center}
	\vspace{-5ex}
\end{figure}
\subsection{Trading Off Computation and Communication Latencies}\label{sec_time} 
In this subsection, we elaborate on the importance of enabling a flexible trade-off between communication load and recovery threshold by analyzing the overall completion time for the matrix multiplication task at hand. The completion delay is the sum of latencies due to computation and communication.

To this end, following a well-established model \cite{lee2017speeding},\cite{aliasgari2018coded}, we assume that computation at each worker $p$ requires a random time $T_p^{\text{comp}}$, measured in some specified unit of time, that is modeled as a shifted exponential distribution with cumulative distribution function (cdf) 
\begin{equation}\label{eq_time}
F^{\text{comp}}(T^{\text{comp}})=1-\exp \left( -\frac{\mu TSD}{tsd}  (T^{\text{comp}}-T_{\min}^{\text{comp}}) \right),
\end{equation}
for $T\geq T_{\min}^{\text{comp}}$ and $F^{\text{comp}}(T)=0$ otherwise. According to \eqref{eq_time}, the parameter $T_{\min}^{\text{comp}}$
represents the minimum processing time, and $1/\mu$ represents the average excess computing time, with respect to $T_{\min}^{\text{comp}}$, per multiplication (recall Remark \ref{remark2}).
Assuming independent computing times, for a given recovery threshold $P_R$, the computation time $T^{\text{comp}}$ is hence given as the $P_R$th-order statistic, i.e., the $P_R$th smallest variable, among the i.i.d. variables $(T_{1}^{\text{comp}},\ldots,T_{P}^{\text{comp}})$.   
Its expectation is given by \cite{ross2014introduction}
\begin{figure}[t!]
	\begin{center}
		\includegraphics[scale=0.6]{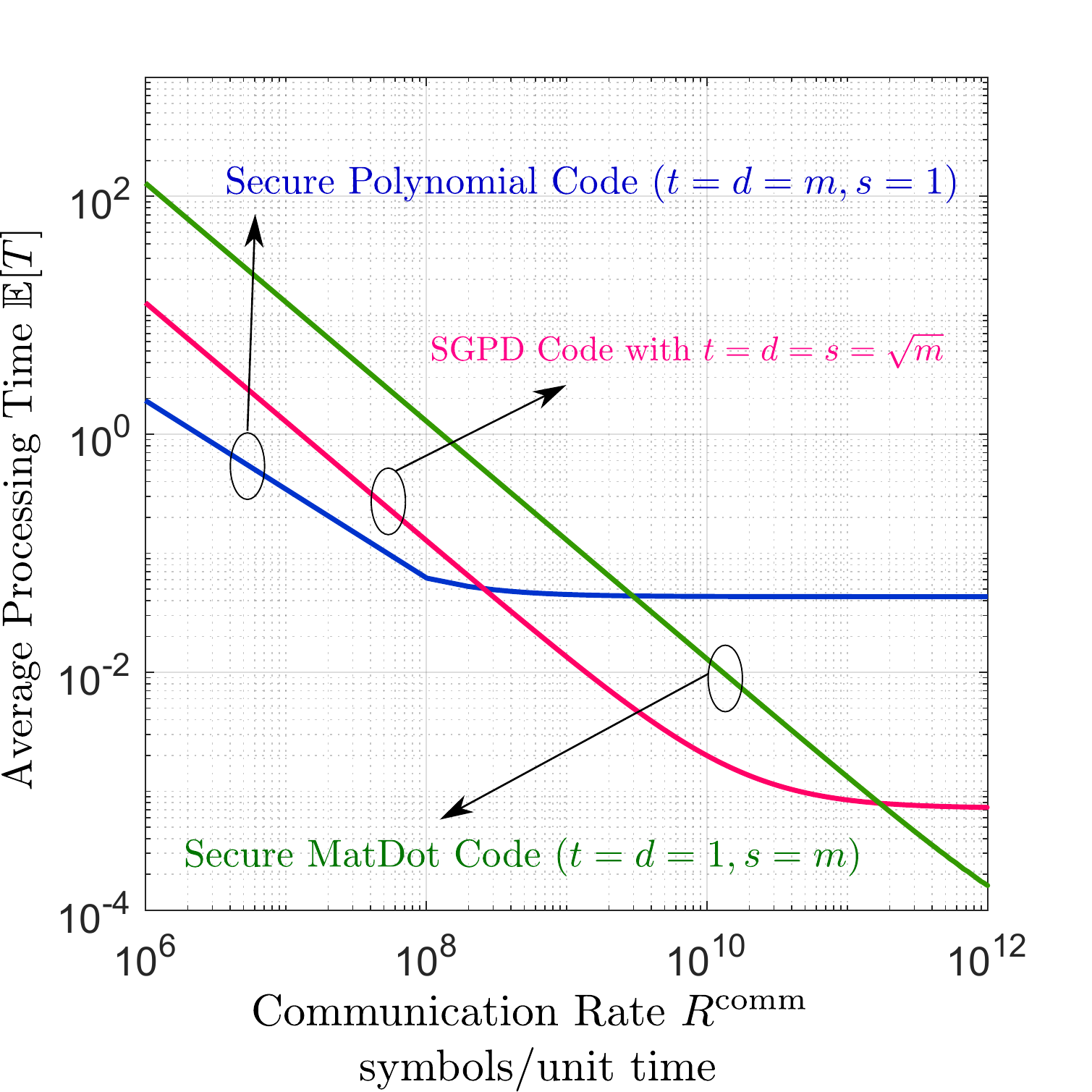}\vspace{-.01cm}~\caption{\footnotesize{Average completion time $\mathbb{E}[T]$ versus communication rate $R^{\text{comm}}$ for secure generalized PolyDot (SGPD) codes with $P=3000$, $P_C=29$, $T=S=D=1008$, $\mu=0.5\times 10^{-4}$, and $T^{\text{comp}}=1$, and $m=n=36$: \emph{(i)} $t=d=36$, $s=1$ (SGPD code), \emph{(ii)} $t=s=d=6$, and \emph{(iii)} $t=d=1$, $s=36$ (secure MatDot code). }~\label{SMMtime}
		}
	\end{center}
	\vspace{-5ex}
\end{figure}
\begin{equation}
\mathbb{E}[T^{\text{comp}}]\!=\!\frac{tsd}{\mu TSD}\sum_{i=1}^{P_R}\frac{1}{P\!-\!P_R\!+\!i}\!=\!\frac{tsd}{\mu TSD }(H_P-H_{P\!-\!P_R}),
\end{equation}
where $H_P$ is the generalized harmonic number defined as $H_P=\sum_{i=1}^{P}1/i$.

Suppose now that the workers communicate with the master server are a link with an overall download rate $R^{\text{comm}}$ (symbols per unit time). The communication latency is hence given as 
\begin{equation}
T^{\text{comm}}=P_R\frac{TD}{tdR^{\text {comm}}},
\end{equation}
since the workers need to return $P_RTD/(td)$ symbols to the master server.
Overall, the average completion time  is given as 
\begin{equation}\label{eq_completionTime}
\mathbb{E}[T]=T_{\min}^{\text{comp}}+\frac{tsd}{\mu TSD }(H_P-H_{P-P_R})+P_R\frac{TD}{ tdR^{\text {comm}}}.
\end{equation}
\begin{Example}
Let consider $P=3000$ workers and parameters $m=n=36$. We assume that $P_C=29$, $T=S=D=1008$, $\mu=0.5\times10^{-4}$, and $T_{\min}^{\text{comm}}=1$. We compare the performance of the following SGPD codes: \emph{(i)} $t=d=36$ and $s=1$ (secure Polynomial code); \emph{(ii)} $t=s=d=6$; {\emph{(iii)}} $t=d=1$ and $s=36$ (secure MatDot code).
The values of $C_L$ and $P_R$ for these codes are shown in Fig.~\ref{figSecRT_CO}.
The average completion time \eqref{eq_completionTime} is plotted versus the communication rate $R^{\text{comm}}$ in Fig.~\ref{SMMtime}.
The figure shows that the optimal choice of the latency-minimizing SGPD code along the curve in Fig.~\ref{figSecRT_CO} depends on the system's operating point: For small communication rates, it is preferable to reduce the communication load $C_L$, and hence secure Polynomial codes are the best choice; while for large communication rate, it is optimal to choose codes with an increasingly large value of the communication load $C_L$.
\end{Example}

\section{Secure and Private Generalized PolyDot Code}\label{Sec_sec_Pri}
   
In this section, we study the setup shown in Fig.~\ref{FigSyst}. We propose a variant of the private and secure GPD code introduced in \cite{kim2019private} that we refer to as private and secure GPD (PSGPD) code. Note that in \cite{kim2019private} a private coded matrix multiplication scheme is proposed only for Polynomial codes with $s=1$ in \eqref{Polydot}. 
We derive the corresponding achievable set of pairs $(P_R,C_L)$ as defined in Section \ref{secModel} under the condition $P_C=1$, i.e., the workers do not collude.

 \begin{Theorem}\label{Thm3PSGPD}
For a given security level $P_C=1$, there is an achievable PSGPD codes with the recovery threshold 
\begin{equation}\label{eq_RT3}
P_R = 
\begin{cases}
s(t+1)d,& \text{ if } s<t,\\
ts(d+1)-t+1 ,&  \text{ if } s\geq t,
\end{cases}
\end{equation}
and the communication load \eqref{eq_CL}, 
for any integer values $t,s$, and $d$ such that $m=ts$, and $n=sd$.
\end{Theorem}
\begin{proof} 
The proof is presented in Appendix \ref{app}.
\end{proof}
\begin{Remark}
The computational complexity of PSGPD codes for both workers and master server is summarized as follows. 	
In PSGPD codes, each worker has two duties, namely encoding the library $\mathcal B$ and computing the multiplication $\mathbf C_p^{(\kappa)}=\mathbf A_p^{(\kappa)}\mathbf B^{(\kappa)}_p$. Encoding the library, i.e., computing the matrix $\mathbf B_p^{(\kappa)}$ in \eqref{eqB}, requires to evaluate $\mathbf F_{\mathbf B^{(r)}}(z)$, $r\in[1,L]$ at query vector $\mathbf q_p^{(\kappa)}$. Hence, the former task requires $LSD$ multiplications, while the latter entails $TSD/(tsd)$ multiplications. In total, each worker carries out $LSD+TSD/(tsd)$ multiplications.
The master server encodes matrix $\mathbf A_p^{(\kappa)}$ with $(1+ts)TS/(ts)$ multiplications. In total, for all $P$ workers, the master server needs $P(1+ts)TS/(ts)$ multiplications.
The computation complexity of the decoding complexity of the master server is the same as for SGPD codes, namely $\mathcal O((P_R-1)(\log(P_r-1))^2TD/(td)))$. 
\end{Remark}
\begin{Example}
Let us consider $P=3000$ workers and parameters $m=n=36$. We assume that $P_C=1$ in order to compare the performance of proposed SGPD and PSGPD codes. Note that both recovery threshold and communication load of the PSGPD code do not depend on the number of public matrices  $|\mathcal B|=L$ in the library. The trade-off between communication load $C_L$ and recovery threshold $P_R$ is illustrated in Fig.~\ref{figpriSecRT_CO} for both codes. The figure shows that, for a fixed value of $P_R$, the resulting achievable value of the communication load $C_L$ is smaller for PSGPD than for SGPD codes.
This suggests that the privacy requirement on the index $\kappa$ imposed by PSGPD is less demanding than the security constraint on matrix $\mathbf B$ under which SGPD codes operate.
\end{Example}
\begin{Remark}
	As for SGPD codes, the overall average completion time of PSGPD codes can be derived following the same steps as described in Section \ref{sec_time}.
\end{Remark}


\begin{figure}[t!]
	\begin{center}
		\includegraphics[scale=0.64]{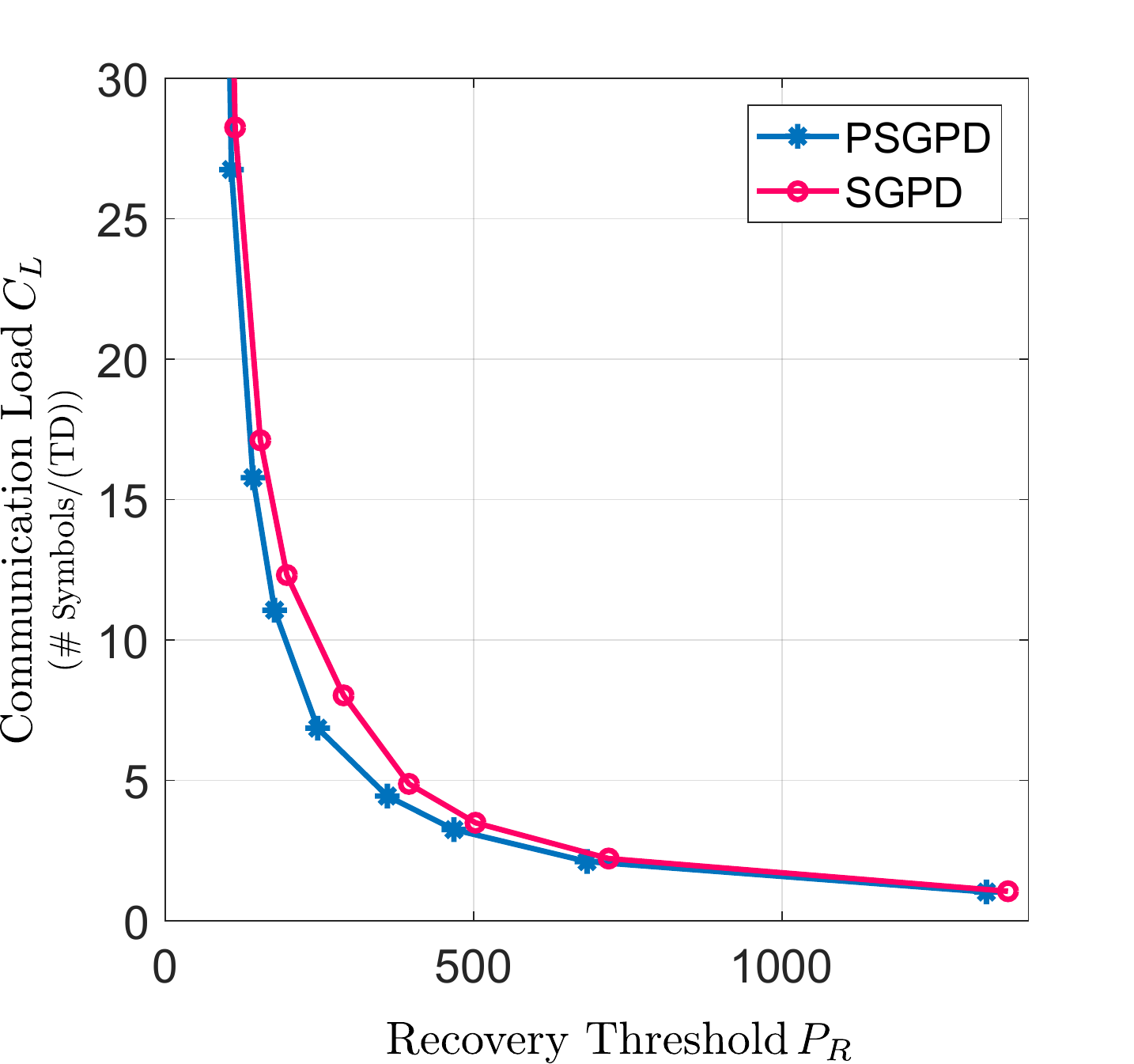}\vspace{-.1cm}~\caption{\footnotesize{Communication
				load $C_L$ versus recovery threshold $P_R$ for 		secure generalized PolyDot (SGPD) codes with $P_C=1$ and private and secure 	generalized PolyDot (PSGPD) codes ($m=n=36$ and $P=3000$ workers).
		}}~\label{figpriSecRT_CO}
	\end{center}
	\vspace{-5ex}
\end{figure}
\section{Concluding Remarks}\label{Sec_Con}
 
In this work, we have considered the problem of secure and private distributed matrix multiplication on $\mathbf C=\mathbf {AB}$ in terms of design of computational codes for two settings. In the first setting, the two matrices $\mathbf A$ and $\mathbf B$ contain confidential data and must be kept secure from the workers; and in the second setting
, matrix $\mathbf A$ is confidential, while matrix $\mathbf B$ is selected in a private manner from a library of public matrices.
For both problems, this work presents the best currently known trade-off between communication load and recovery threshold. This is done by presenting two code constructions that generalize the state-of-the-art GPD codes  \cite{fahim2017optimal,dutta2018optimal,dutta2018unified}, in combination with PIR based codes \cite{kim2019private}. 

Among important items for future research, we mention the extension of the proposed PSGPD construction to $P_C>1$. Here, we note that one can design an achievable PSGPD scheme for any arbitrary privacy level by trivially concatenating a robust PIR scheme for arbitrary colluding workers and private databases \cite{sun2017capacity} and the proposed SGPD code. However, this approach would require multiplying the data matrix $\mathbf A$ with all $L$ public matrices in the set $\mathcal B=\{\mathbf B^{(r)}\}_{r=1}^L$ for each worker $p\in [1,P]$, implying a significantly increased computation load.
Future work will focus on PSGPD schemes for any number of colluding workers that provides a smaller computational complexity at the workers. Finally, the establishment of a converse bound and the consideration of non-perfect communication channels between workers and master server \cite{ha2018wireless} are open problems.

	\appendices
	
	\section{Proof of Theorem \ref{Thm3PSGPD}}\label{app} 

	We start by discussing the $s<t$ case, as done in Section \ref{sec_SecPDC}. The polynomial encoding function for the input matrix $\mathbf A$, is obtained is defined as in \eqref{eqsec1_pAM1li} for $P_C=1$, that is
	\begin{align} 
	\mathbf{F}_{\mathbf{A}}(z)&
	=\sum_{i=1}^{t}\sum_{j=1}^{s} \mathbf{A}_{i,j}z^{s(i-1)+(j-1)}+\mathbf{R}  z^{st},\label{eqa}
	\end{align}
	where we recall that $\mathbf{R}$ is an $T/t\times S/s$ random matrix with i.i.d. uniform random elements in $\mathbb F$. The encoded matrices are given as $\mathbf A_p^{(\kappa)}=\mathbf F_{\mathbf A}(z_{\kappa,p})$ for values $z_{\kappa,p}$ to be discussed below.
	For the desired index $\kappa$, the master server also computes the query vector $\mathbf{q}_p^{(\kappa)}$ for all $p\in [1,P]$. This is obtained as
	\begin{equation}\label{eq_query}
	\mathbf q_p^{(\kappa)}=[z_{1},\ldots,z_{\kappa-1},z_{\kappa,p},z_{\kappa+1},\ldots,z_{L}],
	\end{equation} 
	where all points $\{z_i\}_{i\neq \kappa}$ are selected uniformly i.i.d. from $\mathbb F$ but are identical for all $p$. The points $\{z_{\kappa,p}\}_{p=1}^P$ are selected i.i.d. as distinct elements from $\mathbb F$ (recall that we have $|\mathbb F|>P)$. We note that, as in the PIR scheme \cite{kim2019private}, the query vector \eqref{eq_query} does not leak any information on index $\kappa$ in the sense defined by condition \eqref{eq_privacy}.
	The master server evaluates $\mathbf F_{\mathbf A}(z)$ in \eqref{eqa} at the distinct random point $z_{\kappa,p}$, to produce the encoded matrices $\mathbf A_p^{(\kappa)}=\mathbf F_{\mathbf A}(z_{\kappa,p})$, and then sends $\mathbf A_p^{(\kappa)}$ along with the query vector $\mathbf q_p^{(\kappa)}$ to worker $p\in[1,P]$.   
	
	Each worker $p$, after receiving the query vectors $\mathbf q_p^{(\kappa)}$, encodes the library $\mathcal B$ into a matrix $\mathbf B_p^{(\kappa)}$ as follows. 
	Define the polynomial encoding function for each matrix $\mathbf B^{(r)}$, $r\in [1,L]$, in the library $\mathcal B$ as in \eqref{eqsec1_pBM1li} for $P_C=0$, i.e.,
	\begin{align}
	\mathbf{F}_{\mathbf{B}^{(r)}}(z)&
	= \sum_{k=1}^{s}\sum_{l=1}^{d} \mathbf{B}_{k,l}^{(r)}z^{s-k+(l-1)s(t+1)}.\label{eq_1_B}
	\end{align}
	Each worker $p$ computes the encoded matrices as
	\begin{align}\label{eqB}
	\mathbf{B}_p^{(\kappa)}\overset{\Delta}{=}&
	\sum _{r\in[1,L]}\mathbf{F}_{\mathbf{B}^{(r)}}([\mathbf q^{(\kappa)}_p]_r)\nonumber\\
	=&
	\mathbf{F}_{\mathbf{B}^{(\kappa)}}(z_{\kappa,p})+\sum_{r\in[1,L]\setminus \kappa} \mathbf{F}_{\mathbf{B}^{(r)}}(z_{r}),
	\end{align}
	where $[\mathbf q^{(\kappa)}_p]_r$ denotes the $r$th element of the query vector $\mathbf q^{(\kappa)}_p$. 
	
	After encoding the library, each worker $p$ computes the matrix product $\mathbf C_p^{(\kappa)}=\mathbf A_p^{(\kappa)} \mathbf B_p^{(\kappa)}$ and then sends $\mathbf C _p^{(\kappa)}$ back to the master server.
	We note that both polynomials $\mathbf{F}_{\mathbf{A}}(z)$ and $\mathbf{F}_{\mathbf{B}^{(\kappa)}}(z)$, assigned to the input matrix $\mathbf A$ and the desired matrix $\mathbf B^{(\kappa)}$, are evaluated at the same random points $z_{\kappa,1},\ldots,z_{\kappa,P}$ for workers $1,\ldots,P$, respectively. Since each undesired matrix is evaluated at an identical random point for all workers the second term in \eqref{eqB}, i.e., $\sum_{r\in[1,L]\setminus \kappa} \mathbf F_{\mathbf B^{(r)}}(z_r)$, can be considered as a constant term. 
	
	To reconstruct all blocks $\mathbf C_{i,l}^{(\kappa)}$ of the product matrix $\mathbf C^{(\kappa)}=\mathbf {AB}^{(\kappa)}$, the master server carries out polynomial interpolation, upon receiving a number of
	multiplication results equal to at least $\deg(\mathbf F_{\mathbf A}(z)\mathbf G_{\mathbf B^{(\kappa)}}(z)) + 1$, which is $s(t+1)d$, for the case $s<t$.

	Similarly, for the $s\geq t$ case, the polynomial encoding function for the input matrix $\mathbf A$ as in \eqref{eqsec2_pAM1li} for $P_C=1$, that is, 
	\begin{align}
	\mathbf{F}_{\mathbf{A}}(z)&
	=\sum_{i=1}^{t}\sum_{j=1}^{s} \mathbf{A}_{i,j}z^{ i-1+t(j-1)}+\mathbf{R}z^{ts} ,\label{eqaa}
	\end{align}
	and the encoding function for matrices $\mathbf B^{(r)}$ is given as in  \eqref{eqsec2_pBM1li} for $P_C=0$, that is
	\begin{align}
	\mathbf{F}_{\mathbf{B}^{(r)}}(z)&
	= \sum_{k=1}^{s}\sum_{l=1}^{d} \mathbf{B}_{k,l}^{(r)}z^{(s-k)t+ts(l-1)}.
	\end{align}
	The encoded matrices $\mathbf A_p^{(\kappa)}$ and $\mathbf B_p^{(\kappa)}$ are defined as above, and so are the query vectors $\mathbf q_p^{(\kappa)}$ for all $p\in[1,P]$.
	
	The security of the data matrix $\mathbf{A}$ against non-colluding workers is guaranteed by appending the random matrix $\mathbf{R}$ to the input matrix $\mathbf{A}$ in \eqref{eqa} in the same way as described in Section \ref{sec_SecPDC}. The details for both cases $s<t$ and $s\geq t$ are given in the proofs of Theorems \ref{SecurThm} and \ref{SecurThm2}, respectively, for the case of $P_C=1$. The privacy condition of \eqref{eq_privacy} follows by definition of the query vectors \eqref{eq_query} for the desired index ${\kappa}\in [1,L]$, as proved in \cite{kim2019private}.
	Finally, the recovery threshold and the communication load follow in a manner analogous to Theorems \ref{SecurThm} and \ref{SecurThm2}.


  \ifCLASSOPTIONcaptionsoff
 	\newpage
 	\fi


 	
 	%
 	
 	
 	 \balance
 	\bibliographystyle{IEEEtran} 
 	\bibliography{IEEEabrv,references} 

\begin{thebibliography}{10}
\providecommand{\url}[1]{#1}
\csname url@samestyle\endcsname
\providecommand{\newblock}{\relax}
\providecommand{\bibinfo}[2]{#2}
\providecommand{\BIBentrySTDinterwordspacing}{\spaceskip=0pt\relax}
\providecommand{\BIBentryALTinterwordstretchfactor}{4}
\providecommand{\BIBentryALTinterwordspacing}{\spaceskip=\fontdimen2\font plus
\BIBentryALTinterwordstretchfactor\fontdimen3\font minus
  \fontdimen4\font\relax}
\providecommand{\BIBforeignlanguage}[2]{{%
\expandafter\ifx\csname l@#1\endcsname\relax
\typeout{** WARNING: IEEEtran.bst: No hyphenation pattern has been}%
\typeout{** loaded for the language `#1'. Using the pattern for}%
\typeout{** the default language instead.}%
\else
\language=\csname l@#1\endcsname
\fi
#2}}
\providecommand{\BIBdecl}{\relax}
\BIBdecl

\bibitem{aliasgari2019distributed}
M.~Aliasgari, O.~Simeone, and J.~Kliewer, ``Distributed and private coded
  matrix computation with flexible communication load,'' in \emph{Proc. IEEE
  Intern. Symp. Inform. Theory (ISIT)}, Jul. 2019, pp. 1092--1096.

\bibitem{janzamin2015beating}
M.~Janzamin, H.~Sedghi, and A.~Anandkumar, ``Beating the perils of
  non-convexity: Guaranteed training of neural networks using tensor methods,''
  \emph{arXiv preprint, arXiv:1506.08473}, 2015.

\bibitem{li2014scaling}
M.~Li, D.~G. Andersen, J.~W. Park, A.~J. Smola, A.~Ahmed, V.~Josifovski,
  J.~Long, E.~J. Shekita, and B.-Y. Su, ``Scaling distributed machine learning
  with the parameter server.'' in \emph{Proc. of the 11th {USENIX} Conference
  on Operating Systems Design and Implementation, OSDI}, vol.~14, Oct. 2014,
  pp. 583--598.

\bibitem{dean2013tail}
J.~Dean and L.~A. Barroso, ``The tail at scale,'' \emph{Communications of the
  ACM}, vol.~56, no.~2, pp. 74--80, Feb. 2013.

\bibitem{joshi2017efficient}
G.~Joshi, E.~Soljanin, and G.~Wornell, ``Efficient redundancy techniques for
  latency reduction in cloud systems,'' \emph{ACM Transactions on Modeling and
  Performance Evaluation of Computing Systems (TOMPECS)}, vol.~2, no.~2, pp.
  12:1--12:30, Apr. 2017.

\bibitem{wang2015using}
D.~Wang, G.~Joshi, and G.~Wornell, ``Using straggler replication to reduce
  latency in large-scale parallel computing,'' \emph{ACM SIGMETRICS Performance
  Evaluation Review}, vol.~43, no.~3, pp. 7--11, Dec. 2015.

\bibitem{huang1984algorithm}
K.-H. Huang and J.~A. Abraham, ``Algorithm-based fault tolerance for matrix
  operations,'' \emph{IEEE Trans. on Computers}, vol. 100, no.~6, pp. 518--528,
  Jun. 1984.

\bibitem{lee2018speeding}
K.~Lee, M.~Lam, R.~Pedarsani, D.~Papailiopoulos, and K.~Ramchandran, ``Speeding
  up distributed machine learning using codes,'' \emph{IEEE Trans. on Inform.
  Theory}, vol.~64, no.~3, pp. 1514--1529, Aug. 2017.

\bibitem{yu2017polynomial}
Q.~Yu, M.~Maddah-Ali, and S.~Avestimehr, ``Polynomial codes: an optimal design
  for high-dimensional coded matrix multiplication,'' in \emph{Proc. Advances
  in Neural Inform. Processing Systems}, Dec. 2017, pp. 4403--4413.

\bibitem{li2018fundamental}
S.~Li, M.~A. Maddah-Ali, Q.~Yu, and A.~S. Avestimehr, ``A fundamental tradeoff
  between computation and communication in distributed computing,'' \emph{IEEE
  Trans. on Inform. Theory}, vol.~64, no.~1, pp. 109--128, Sep. 2017.

\bibitem{aliasgari2018coded}
M.~Aliasgari, J.~Kliewer, and O.~Simeone, ``Coded computation against
  processing delays for virtualized cloud-based channel decoding,'' \emph{IEEE
  Trans. on Commun.}, vol.~67, no.~1, pp. 28--38, Jan. 2019.

\bibitem{aliasgari2018codedisit}
------, ``Coded computation against straggling decoders for network function
  virtualization,'' in \emph{Proc. IEEE Intern. Symp. Inform. Theory (ISIT)},
  Jun. 2018, pp. 711--715.

\bibitem{dutta2018optimal}
S.~Dutta, M.~Fahim, F.~Haddadpour, H.~Jeong, V.~Cadambe, and P.~Grover, ``On
  the optimal recovery threshold of coded matrix multiplication,'' \emph{arXiv
  preprint, arXiv:1801.10292}, 2018.

\bibitem{dutta2018unified}
S.~Dutta, Z.~Bai, H.~Jeong, T.~M. Low, and P.~Grover, ``A unified coded deep
  neural network training strategy based on generalized polydot codes for
  matrix multiplication,'' \emph{arXiv preprint, arXiv:1811.10751}, 2018.

\bibitem{fahim2017optimal}
M.~Fahim, H.~Jeong, F.~Haddadpour, S.~Dutta, V.~Cadambe, and P.~Grover, ``On
  the optimal recovery threshold of coded matrix multiplication,'' in
  \emph{Proc. 55th Allerton Conf. Commun., Control, Comput., IL, USA}, Oct.
  2017, pp. 1264--1270.

\bibitem{fahim2019numerically}
M.~Fahim and V.~R. Cadambe, ``Numerically stable polynomially coded
  computing,'' \emph{arXiv preprint, arXiv:1903.08326}, 2019.

\bibitem{subramaniam2019random}
A.~M. Subramaniam, A.~Heidarzadeh, and K.~R. Narayanan, ``Random
  khatri-rao-product codes for numerically-stable distributed matrix
  multiplication,'' \emph{arXiv preprint, arXiv:1907.05965}, 2019.

\bibitem{nodehi2018limited}
H.~A. Nodehi and M.~A. Maddah-Ali, ``Limited-sharing multi-party computation
  for massive matrix operations,'' in \emph{Proc. IEEE Intern. Symp. on Inform.
  Theory (ISIT)}, Jun. 2018, pp. 1231--1235.

\bibitem{yu2018lagrange}
Q.~Yu, N.~Raviv, J.~So, and A.~S. Avestimehr, ``Lagrange coded computing:
  Optimal design for resiliency, security and privacy,'' \emph{arXiv preprint,
  arXiv:1806.00939}, 2018.

\bibitem{chang2018capacity}
W.-T. Chang and R.~Tandon, ``On the capacity of secure distributed matrix
  multiplication,'' \emph{arXiv preprint, arXiv:1806.00469}, 2018.

\bibitem{kakar2019capacity}
J.~Kakar, S.~Ebadifar, and A.~Sezgin, ``On the capacity and
  straggler-robustness of distributed secure matrix multiplication,''
  \emph{IEEE Access}, vol.~7, pp. 45\,783--45\,799, Apr. 2019.

\bibitem{yang2019secure}
H.~Yang and J.~Lee, ``Secure distributed computing with straggling servers
  using polynomial codes,'' \emph{IEEE Trans. on Inform. Forensics and Secur.},
  vol.~14, no.~1, pp. 141--150, Jan. 2019.

\bibitem{rafael2018codes}
R.~G. D'Oliveira, S.~E. Rouayheb, and D.~Karpuk, ``{GASP} codes for secure
  distributed matrix multiplication,'' \emph{arXiv preprint, arXiv:1812.09962},
  2018.

\bibitem{das2019random}
A.~B. Das, A.~Ramamoorthy, and N.~Vaswani, ``Random convolutional coding for
  robust and straggler resilient distributed matrix computation,'' \emph{arXiv
  preprint, arXiv:1907.08064}, 2019.

\bibitem{nodehi2019secure}
H.~A. Nodehi and M.~A. Maddah-Ali, ``Secure coded multi-party computation for
  massive matrix operations,'' \emph{arXiv preprint, arXiv:1908.04255}, 2019.

\bibitem{ricci2011introduction}
F.~Ricci, L.~Rokach, and B.~Shapira, \emph{Introduction to recommender systems
  handbook}.\hskip 1em plus 0.5em minus 0.4em\relax Springer, 2011.

\bibitem{lee2017high}
K.~Lee, C.~Suh, and K.~Ramchandran, ``High-dimensional coded matrix
  multiplication,'' in \emph{Proc. IEEE Intern. Symp. Inform. Theory (ISIT)},
  Jun. 2017, pp. 2418--2422.

\bibitem{yu2018straggler}
Q.~Yu, M.~A. Maddah-Ali, and A.~S. Avestimehr, ``Straggler mitigation in
  distributed matrix multiplication: Fundamental limits and optimal coding,''
  \emph{arXiv preprint, arXiv:1801.07487}, 2018.

\bibitem{jia2019capacity}
Z.~Jia and S.~A. Jafar, ``On the capacity of secure distributed matrix
  multiplication,'' \emph{arXiv preprint, arXiv:1908.06957}, 2019.

\bibitem{chor1995private}
B.~Chor, O.~Goldreich, E.~Kushilevitz, and M.~Sudan, ``Private information
  retrieval,'' in \emph{Proceedings of IEEE 36th Annual Foundations of Computer
  Science}.\hskip 1em plus 0.5em minus 0.4em\relax IEEE, 1995, pp. 41--50.

\bibitem{gasarch2004survey}
W.~Gasarch, ``A survey on private information retrieval,'' \emph{Bulletin of
  the EATCS}, vol.~82, no. 113, pp. 72--107, Feb. 2004.

\bibitem{yekhanin2010private}
S.~Yekhanin, ``Private information retrieval,'' \emph{Commun. {ACM}}, vol.~53,
  no.~4, pp. 68--73, Apr. 2010.

\bibitem{sun2017capacity}
H.~Sun and S.~A. Jafar, ``The capacity of private information retrieval,''
  \emph{IEEE Trans. on Inform. Theory}, vol.~63, no.~7, pp. 4075--4088, Jul.
  2017.

\bibitem{banawan2018capacity}
K.~Banawan and S.~Ulukus, ``The capacity of private information retrieval from
  coded databases,'' \emph{IEEE Trans. on Inform. Theory}, vol.~64, no.~3, pp.
  1945--1956, Mar. 2018.

\bibitem{Hollanti_etal_2017}
R.~Freij-Hollanti, O.~W. Gnilke, C.~Hollanti, and D.~A. Karpuk, ``Private
  information retrieval from coded databases with colluding servers,''
  \emph{SIAM J. Appl. Algebra Geom.}, vol.~1, no.~1, pp. 647--664, Nov. 2017.

\bibitem{kazemi2019single}
F.~Kazemi, E.~Karimi, A.~Heidarzadeh, and A.~Sprintson, ``Single-server
  single-message online private information retrieval with side information,''
  in \emph{Proc. IEEE Intern. Symp. Inform. Theory (ISIT)}, Jul. 2019, pp.
  350--354.

\bibitem{kazemi2019private}
------, ``Private information retrieval with private coded side information:
  The multi-server case,'' \emph{arXiv preprint, arXiv:1906.11278}, 2019.

\bibitem{kim2019private}
M.~Kim and J.~Lee, ``Private secure coded computation,'' \emph{arXiv preprint,
  arXiv:1902.00167}, 2019.

\bibitem{chang2019upload}
W.-T. Chang and R.~Tandon, ``On the upload versus download cost for secure and
  private matrix multiplication,'' \emph{arXiv preprint, arXiv:1906.10684},
  2019.

\bibitem{tahmasebi2019private}
B.~Tahmasebi and M.~A. Maddah-Ali, ``Private sequential function computation,''
  \emph{arXiv preprint, arXiv:1908.01204}, 2019.

\bibitem{shamir1979share}
A.~Shamir, ``How to share a secret,'' \emph{Communications of the {ACM}},
  vol.~22, no.~11, pp. 612--613, Nov. 1979.

\bibitem{kung2009fast}
H.-T. Kung, \emph{Fast evaluation and interpolation}.\hskip 1em plus 0.5em
  minus 0.4em\relax Carnegie Mellon University, Tech. Rep., 2009.

\bibitem{lee2017speeding}
K.~Lee, M.~Lam, R.~Pedarsani, D.~Papailiopoulos, and K.~Ramchandran, ``Speeding
  up distributed machine learning using codes,'' \emph{IEEE Trans. on Inform.
  Theory}, vol.~64, no.~3, pp. 1514--1529, Mar. 2017.

\bibitem{ross2014introduction}
S.~M. Ross, \emph{Introduction to Probability Models}.\hskip 1em plus 0.5em
  minus 0.4em\relax Academic {P}ress, 2014.

\bibitem{ha2018wireless}
S.~Ha, J.~Zhang, O.~Simeone, and J.~Kang, ``Coded federated computing in
  wireless networks with straggling devices and imperfect {CSI},'' in
  \emph{Proc. IEEE Intern. Symp. Inform. Theory (ISIT)}, Jul. 2019, pp.
  2649--2653.

\end{thebibliography}

 \end{document}